%% file: Project_main.tex
\documentclass[12pt,twoside,a4paper]{article}

\input{Preamble/packages}

\input{Preamble/commands}

\begin{document}

\title{\Huge Exact Parallel Waves in General Relativity}
\author{Cian Roche$^{*,\dag}$, Amir Babak Aazami$^\ddag$, Carla Cederbaum$^\dag$}

\maketitle
\thispagestyle{empty}

\begin{abstract}
We conduct a review of the basic definitions and the principal results in the study of wavelike spacetimes, that is spacetimes whose metric models massless radiation moving at the speed of light, focusing in particular on those geometries \textit{with parallel rays}. In particular, we motivate and connect their various definitions, outline their coordinate descriptions and present some classical results in their study in a language more accessible to modern readers, including the existence of ``null coordinates" and the construction of Penrose limits. We also present a thorough summary of recent work on causality in pp-waves, and describe progress in addressing an open question in the field -- the Ehlers--Kundt conjecture. 
\end{abstract}

\vfill
{\footnotesize 
$^*$Department of Physics, Massachusetts Institute of Technology, Cambridge, MA, USA\\
$^\dag$Department of Mathematics, Eberhard Karls Universit\"at T\"ubingen, Germany\\
$^\ddag$Department of Mathematics, Clark University, Worcester, MA, USA\\
Corresponding author: Cian Roche, \href{mailto:roche@mit.edu}{roche@mit.edu}}

\newpage

\setcounter{page}{1}

\newpage

\input{Introduction}

\input{Defining Waves in GR/linearised}

\input{Defining Waves in GR/defining}

\input{coordDesc}

\input{Properties/invariants}

\input{Properties/wavefront}

\input{Properties/PenroseLimit}

\input{Properties/Causal}

\input{EKConj}

\section*{Acknowledgements}
The authors wish to thank Prof. Miguel S\'anchez for numerous helpful discussions, and Prof. Paweł Nurowski for valuable clarifications. We also thank Luke Timmons and John Walker for their continued support and assistance. Finally, we thank the anonymous reviewers who provided numerous invaluable suggestions.

\section*{Data Availability}
Data sharing not applicable to this article as no datasets were generated or analysed during the current study.

\printbibliography
\appendix
\input{Postamble/appendix}

\end{document}

%% file: Preamble/packages.tex
\usepackage[a4paper,bindingoffset=0.2in,%
            left=1in,right=1in,top=1in,bottom=1in,%
            footskip=.25in]{geometry}
\usepackage{graphicx}
\usepackage{amsmath}  
\usepackage{verbatim}     
\usepackage{hyperref}
\usepackage[noabbrev]{cleveref}  
\usepackage{mwe}
\usepackage[mathscr]{euscript} 
\usepackage{caption}
\usepackage{subcaption}
\usepackage{amssymb}
\usepackage{amsthm}
\let\oldproofname=\proofname
\renewcommand{\proofname}{\rm{\oldproofname}}
\usepackage{mathtools}
\usepackage{dsfont}
\usepackage{nicematrix}
\NiceMatrixOptions{nullify-dots}
\usepackage[shortlabels]{enumitem}
\usepackage{musicography}
\usepackage[hyperref=true,
    backend=biber,
    sorting=none,
    style=numeric-comp
    ]{biblatex}
    
\usepackage[titletoc]{appendix}
\usepackage{xcolor}
\usepackage{longtable}
\usepackage{cancel}
\usepackage{lscape}
\usepackage{tikz}

%% file: Preamble/commands.tex
\newcommand{\reals}{{\rm I\!R}}
\newcommand{\munu}{{\mu\nu}}
\newcommand\comma{\hspace{.2in},\hspace{.2in}}

\renewcommand{\flat}{{\text{\musFlat{}}}}
\renewcommand{\sharp}{{\text{\musSharp{}}}}

\newtheorem{theorem}{Theorem}[section]
\newtheorem{proposition}[theorem]{Proposition}
\newtheorem{Definition}{Definition}
\newtheorem{Lemma}[theorem]{Lemma}
\theoremstyle{definition}
\newtheorem{Remark}[theorem]{Remark}

\makeatletter
\newcommand*{\defeq}{\mathrel{\rlap{%
                     \raisebox{0.3ex}{$\m@th\cdot$}}%
                     \raisebox{-0.3ex}{$\m@th\cdot$}}%
                     =}
\makeatother
\newcommand\cd[2]{\nabla_{\!#1}{#2}}
\newcommand\conn[2]{\overline{\nabla}_{\!#1}{[#2]}}

%% file: Introduction.tex
\section{Introduction}
The goal of this article is first to make explicit the definitions of wavelike spacetimes with parallel rays in general relativity, and as was done in the development of the theory, to motivate a number of these definitions with reference to the well-established theory of electromagnetism. This is the subject of Section \ref{sec:defining}. By ``wavelike spacetimes" we mean those spacetimes which themselves model wavelike behaviour, in contrast to the spacetimes which model objects that produce radiation\footnote{For such spacetimes, see e.g. the review~\cite{Holst}.}. We then examine the coordinate descriptions of the wavelike spacetimes in Section \ref{sec:coordDesc}, where the ``adapted" or ``Brinkmann" coordinates in which these metrics are typically written are derived. Section \ref{sec:properties} discusses the properties of the wavelike spacetimes, in particular the details of their so-called ``wavefronts", that such a spacetime appears as a limit of any spacetime via the ``Penrose limit", and their causal properties. In Section \ref{sec:EKConj} we discuss progress in addressing the ``Ehlers--Kundt conjecture", which is a statement about our expectations of the wavelike spacetimes based on physical intuition.\\

It should be noted that this article deals only with the wavelike spacetimes which possess \textit{parallel rays}, which is a subclass of all wavelike spacetimes in general relativity. The most general class are those geometries admitting shear-free, twist-free, geodesic
null congruences, which splits into the Kundt class (non-expanding congruence) and the Robinson--Trautman class (expanding congruence). The waves with parallel rays discussed in this article form a subclass of the Kundt class. Other noteworthy classes include the colliding plane waves, cylindrical gravitational waves, spacetimes with accelerated sources (C-metrics, more generally spacetimes with boost-rotation symmetries), solitonic gravitational waves, cosmological gravitational waves in de Sitter and anti-de Sitter spacetimes, exact gravitational waves in FLRW cosmologies, Bianchi cosmologies, and Gowdy universes. For reviews of these topics and modern results other than those presented in the remainder of this article, we direct the reader to the following articles: A summary of the historical development of the mathematics of wavelike exact solutions \cite{summary}, modern references on exact solutions in general relativity \cite{exactSolnsKramer, griffiths_podolský_2009}, works which deal with colliding plane waves and the physical interpretations of certain exact solutions \cite{colliding,Bonnor1994}, and other related reviews \cite{Hawking,MTWGrav,Carmeli,solitons}.

\subsection{Survey of Early Developments}
We begin by providing a brief historical perspective on the development of the theory of waves in general relativity (GR) as in \cite{gravwave} with some relevant additions. In particular, we outline here \textbf{only the early results} in the field in order to supplement the material of Sec. \ref{sec:defining}, and leave discussion of modern developments not covered elsewhere in this article to the references above.

\begin{longtable*}{rp{0.9\textwidth}}
1915 \quad &Albert Einstein establishes the field equation of general relativity\\
1916 \quad &Einstein demonstrates that the linearised vacuum field equation admits wavelike solutions which are rather similar to electromagnetic waves\\
1918 \quad &Einstein derives the quadrupole formula according to which gravitational waves are produced by a time-dependent mass quadrupole moment\\
1925 \quad &Hans Brinkmann finds a class of exact wavelike solutions to the vacuum field equation, later called \textbf{pp}-waves (``\textbf{p}lane-fronted waves with \textbf{p}arallel rays") by J\"urgen Ehlers and Wolfgang Kundt. Note that this was a purely mathematical work, and they were not yet understood as modelling massless radiation.\\
1926 \quad &Baldwin and Jeffery illuminate the interpretations of wavelike spacetimes when amplitudes are not assumed to be small \cite{baldwin}\\
1936 \quad &Einstein submits, together with Nathan Rosen, a manuscript to Physical Review in which they claim that gravitational waves do not exist\\
1937 \quad &After receiving a critical referee report from Howard
P. Robertson, Einstein withdraws the manuscript with the erroneous claim and publishes, together with Rosen, a strongly revised manuscript on wavelike solutions (Einstein-Rosen waves) in the Journal of the Franklin Institute\\
1957 \quad &Felix Pirani gives an invariant (i.e. coordinate-independent) characterisation of gravitational radiation, and Bondi independently writes down a metric for the plane wave which is singularity-free and carries energy \cite{Bondi1957}. This work was later developed by Asher Peres in 1959 \cite{Peres}\\
1958 \quad &Anderzej Trautman  reformulates Sommerfeld’s radiation boundary conditions for a general field theory, and applies this approach to relativity to find the boundary conditions to be imposed at infinity due to bounded sources in GR\\
1960 \quad &Ivor Robinson and Trautman discover a class of exact solutions to Einstein's vacuum field equation that describe outgoing gravitational radiation\\
1961 \quad &Wolfgang Kundt surveys the wavelike geometries as those admitting a twistfree and non-expanding null congruence, and characterizes their subclasses of different Petrov type by
geometrical properties \cite{Kundt1961,Kundt1962}\\
1962 \quad &Ehlers and Kundt conjecture that the gravitational pp-waves other than the plane wave cannot be complete\\
1962 \quad &Roger Penrose provides a geometric definition of asymptotic flatness, along with various new studies of the asymptotic properties of spacetimes including definitions and conservation laws for energy and momentum\\
1965 \quad &Penrose shows that the plane waves (gravitational or otherwise) are not globally hyperbolic\\
1976 \quad &Penrose demonstrates a limiting procedure by which any spacetime reduces to a plane wave, by ``blowing up" a neighbourhood of a null geodesic
\end{longtable*}

In the remainder of this article, we detail a subset of these results followed by a selection of advances of the theory that have taken place in the decades since. Again for details on modern advances not within the scope of this article, see \cite{summary,exactSolnsKramer,griffiths_podolský_2009,colliding,Bonnor1994,Hawking,MTWGrav,Carmeli,solitons}.



\subsection{Nomenclature}
The names used to refer to different classes of wavelike geometries in this article are not all standard in the literature, due to a degree of degeneracy in the usage of certain terms; eg. ``pp-wave'' can implicitly refer to a 4-dimensional geometry with \textit{planar} wavefront, or to an $n$-dimensional geometry with \textit{curved} wavefront. Also sometimes ambiguous is the local or global nature of coordinates used in the description of wavelike geometries. Due to the importance of dimension, global characteristics and wavefront geometry in determining the properties of the wave, the authors see it as necessary to fix one consistent language for the purposes of this article. To summarize these definitions and to facilitate comparison with the literature, we fix nomenclature in Table \ref{tab:definitions} below. In particular, note that the term ``parallel wave" has not previously been used, and instead the term ``pp-wave" is often used in the literature to refer to the same object with the understanding that the geometry in question need not have planar wavefront\footnote{The ``wavefront" of a parallel wave is defined precisely in Def. \ref{def:wavefront}.}.

\section{Defining Waves in General Relativity}\label{sec:defining}
Let us now set about attempting to define a wave in GR. This is not a simple task because of the inherent nonlinearity of GR, and so we take inspiration from the well-established \textit{linear} wave theory of electromagnetism. To this end, we will start by looking at linearised/``weak-field" GR, and demonstrate that in this linear regime one finds wavelike behaviour analogous to Maxwell's electromagnetism (EM), with some fundamental differences. Such differences have origin\footnote{In fact there is another major difference, which is that there are two signs of charge in electromagnetism and only one in gravitation. Such a property is very relevant in fields like cosmology (EM fields are screened but gravitational fields are not), but will not impact our attempts at defining wavelike behaviour.} in the fact that the relevant field object in EM is a 1-tensor (the vector potential $A^\mu$) whereas in GR the relevant object is a 2-tensor (the metric $g_\munu$).\\

Upon finding such behaviour in the linear regime, we will discuss how to extend the results to the general case. This will be accomplished by taking inspiration from the covariant properties of the linearised waves (those properties which do not depend on the coordinate system used), and showing that a general metric satisfying such properties exhibits similar wavelike behaviour.

\newgeometry{left=2cm, right=2cm, bottom=2cm, top=2cm}
\begin{landscape}
\begin{table}[]
    \small
    \renewcommand{\arraystretch}{1.3}
    \linespread{0.95}
    \begin{tikzpicture}
        \node[draw,ultra thick,inner sep=.4pt]{\begin{tabular}{p{0.13\linewidth} | p{0.67\linewidth} | p{0.18\linewidth}}
        In this article & Definition                                                                                                                                                                                                                                                                                                                                                                                                                                                                                                                              & Other used names                                                   \\ \hline
        Parallel wave        & A non-flat Lorentzian manifold which admits a global, covariantly constant null vector field  (Def. \ref{def:parallelRays}).                                                                                                                                                                                                                                                                                                                                                                                                                                          & pp-wave (understood \textbf{not} to refer to any planar character)                                                                      \\ \hline
        pp-wave              & A non-flat Lorentzian manifold $(M,g)$ which admits a global, covariantly constant vector field $Z$, for which the curvature tensor $R$ additionally satisfies $R\vert_{Z_\perp \wedge Z_\perp}=0$ where $Z_\perp \vcentcolon= \{X\in TM ~|~ g(X,Z)=0\}$. Informally, a parallel wave with flat wavefront (Def. \ref{def:pp}). \newline Typically, a local gauge freedom is exploited such that some off-diagonal (or the ``gyratonic") terms are omitted when the metric is written in Brinkmann coordinates.             & Brinkmann space, \newline plane-fronted wave                                                         \\ \hline
        Gyratonic pp-wave    & These are pp-waves which are explicitly written with off-diagonal terms present in Brinkmann coordinates, representing a ``spinning" character of the source (Sec. \ref{sec:gyratonic}). \newline Note that the gyratonic pp-waves can also be studied with non-flat wavefront, in which case a more appropriate name would be ``gyratonic parallel wave" to emphasize the gyratonic character.                                                                                                                                   &                                                                         \\ \hline
        Standard pp-wave     & A pp-wave for which the so-called ``Brinkmann coordinates" exist globally, and one can omit the ``gyratonic" terms without losing global information (Sec. \ref{sec:standardPp}).                                                                                                                                                                                                                                                                                                                                         & Classical pp-wave, \newline pp-wave                                              \\ \hline
        Classical pp-wave    & A standard pp-wave in 4 dimensions (Sec. \ref{sec:classicalPp}).                                                                                                                                                                                                                                                                                                                                                                                                                                                                                              & Standard pp-wave, \newline pp-wave                                               \\ \hline
        $(N,h)$p-wave        & Informally, this is a pp-wave for which the flat wavefront is replaced by a ``constant" Riemannian manifold, where constant refers to the fact that the components of the metric on the wavefront are independent of $u$ when written in Brinkmann coordinates (Sec. \ref{sec:Nhp}).                                                                                                                                                                                                                               & pp-wave,\newline N-fronted wave,\newline Generalized plane wave,\newline Plane-fronted wave \\ \hline
        Plane wave           & \begin{enumerate}[noitemsep, leftmargin=*, topsep=0pt, before={\vspace*{-0.5\baselineskip}}]\item A non-flat Lorentzian manifold which admits a 5-parameter group of isometries (Def. \ref{def:planeSymmetry}) \item A classical pp-wave for which the coefficient $H$ in Brinkmann coordinates is \textit{quadratic} in $x,y$ (Sec. \ref{sec:plane}) \item A classical pp-wave defined by a covariantly constant null vector field $Z$ which additionally satisfies $\nabla_X R = 0 ~\forall~ X \in Z_\perp$ where $R$ is the curvature tensor and $Z_\perp \vcentcolon= \{X\in TM~|~g(X,Z) = 0\}$ (Def. \ref{def:planecurvature}) \end{enumerate} &      \\ \hline
        Sandwich wave        & A plane wave with compactly supported curved region, such that the characteristic function $H$ in Brinkmann coordinates satisfies $H(u,x,y)=0$ unless $u \in (a,b) \subset \mathbb{R}$. In the limit of shrinking support $(a,b)$ one obtains the so-called ``impulsive waves''.                                                                                                                                                                                                                         &        \\ \hline                                                               
    \end{tabular}};
    \end{tikzpicture}
    \caption{\label{tab:definitions}Nomenclature summary, covering the definitions of this article and some terminology which has been used to refer to the same objects in the literature. Note that all references to ``Brinkmann coordinates'' refer to the coordinates of Eq. \ref{eq:generalPp}. For further classification of the Ricci-flat classical pp-waves see the table of \cite[pg. 79]{exactsolEK} (wherein the characteristic function $H$ in Brinkmann coordinates is defined with an additional factor of two relative to our notation).}
\end{table}
\end{landscape}
\restoregeometry

%% file: Defining Waves in GR/linearised.tex

\subsection{Linearised Gravity}

Finding wavelike behaviour in the linear/weak-field regime is a very standard calculation, completed first in 1916 by Einstein \cite{einstein1916} but for a modern presentation see for example \cite{gravwave}, \cite{Flanagan_2005}, and \cite{carroll}. As a result, in this section we will only restate the results necessary to build intuition for the later definitions of wavelike behavior. Consider a perturbation $h_\munu$ to the Minkowski background $\eta_\munu$. That is, for the spacetime manifold $M = \reals^4$ we have the Lorentzian metric 
\begin{equation}
    g_\munu = \eta_\munu + h_\munu, ~~~~ |h_\munu| \ll 1
\end{equation}
where we have implicitly chosen local coordinates $x^\mu$, and in these coordinates the Minkowski metric $\eta$ takes the usual form $\text{diag}(-1,+1,+1,+1)$ and the perturbation $h_\munu$ is in some sense ``small". Here, ``smallness" is defined loosely by the fact that the terms quadratic in $h_\munu$ contribute insignificantly to the Einstein equations. We then wish to obtain the Einstein tensor for this metric to linear order in $h_\munu$. To this end, we may raise and lower indices of $h_\munu$ with the background metric $\eta$ since doing so with the full metric $g$ would lead to corrections of order higher than $1$  in $h_\munu$. This can also be viewed as treating the perturbation $h_\munu$ as a symmetric tensor propagating\footnote{One could instead derive the linearised Einstein equations as the equation of motion for $h_\munu$ via a Lagrangian density, however in this article we take the usual approach of calculating the Einstein tensor directly.} on a Minkowski background. For the details of this calculation on a \textit{curved} background, see \cite{gravwave}. \\

To simplify calculations, one chooses to work not with $h_\munu$ but rather with the \textit{trace-reversed} variable $\bar{h}_\munu$ defined as 
\[
    \bar{h}_\munu \vcentcolon= h_\munu -\frac{1}{2} h \eta_\munu
\]
called so because $\bar{h}^\mu{}_\mu = - h^\mu{}_\mu =\vcentcolon -h$ (note that the Einstein tensor is just the trace-reversed Ricci tensor). In electromagnetism one often works with the Lorenz\footnote{Note that this is not a full fixing of the gauge, as the theory remains invariant under transformations of the form $A^\mu \longrightarrow A^\mu + \partial^\mu f$ for a harmonic scalar field $f$. Also note that this gauge goes by many different names in the literature, including (erroneously) the \textit{Lorentz} gauge \cite[footnote pg. 6]{Flanagan_2005}.} gauge conditions $\partial_\mu A^{\mu}=0$ for the vector potential $A^\mu$. Since we are interested in describing radiation in general relativity, we will use the analogous condition 
\begin{equation}\label{eq:LorenzGauge}
    \partial{}^\mu \bar{h}_\munu = 0
\end{equation}
on the trace-reversed perturbation $\bar{h}_\munu$. As a result of these choices, the Einstein tensor is given (to linear order in the perturbation) by \begin{equation}
    G_\munu=-\frac{1}{2}\square \bar{h}_\munu
\end{equation}
where we have defined the D'Alembertian $\square \vcentcolon=\nabla^\mu \nabla_\mu$ which here is simply the flat space D'Alembertian $\square = -\partial_t^2 + \partial_x^2 + \partial_y^2 + \partial_z^2$ (the presence of which is an early sign of wavelike behaviour). Therefore the Einstein equation of linearised gravity reads 
\begin{equation}
    \square \bar{h}_\munu = -16\pi T_\munu
\end{equation}
in units where $c=G=1$ and it is understood that the energy-momentum tensor $T$ is also consistent with the ``weak field" regime. By this, we mean that the lowest nonvanishing order in $T_\munu$ is of the same order of magnitude as the perturbation $h_\munu$. The \textit{vacuum} Einstein equation is then simply a homogeneous wave equation for $\bar{h}_\munu$ and so one makes the plane wave ansatz 
\begin{equation}\label{eq:ansatz}
    \bar{h}_\munu = C_\munu (k) e^{ik_\sigma x^\sigma}
\end{equation}
for some complex, symmetric coefficients $C_\munu$ and $k = (\omega, \mathbf{k})$ a constant vector field on $M$ (constant in the usual sense, since the background metric is flat). As is standard when making a plane wave ansatz written in the complex form, it is understood that at the end of the day, one should take the real part of expressions to obtain physical results.\\

The Lorenz gauge condition Eq. \ref{eq:LorenzGauge} for such a perturbation yields 
\begin{equation}\label{eq:transverse}
k^\mu C_\munu = 0
\end{equation}
for all $\nu$, that is, the perturbation is \textit{orthogonal} to the wave vector. One may interpret this as the fact that a gravitational perturbation of this kind will be \textit{transverse} in a way analogous to the electric and magnetic fields of electromagnetism. The vacuum Einstein equations for such a plane wave perturbation yield
\begin{equation}\label{eq:likeFourier}
    0 = \square \bar{h}_\munu = -k_\sigma k^\sigma \bar{h}_\munu
\end{equation}
which is obtained by noting that $\partial_\sigma \bar{h}_\munu = ik_\sigma \bar{h}_\munu$. Since we are not interested in solutions for which $\bar{h}_\munu$ is identically zero, we instead have 
\begin{equation}
    k_\sigma k^\sigma = 0
\end{equation}
that is, the ``wave vector" $k$ of a plane wave solution to the linearised Einstein equations must be \textit{null}. This is the statement that in the linear theory, the metric exhibits a wavelike behaviour which propagates at the speed of light $c$. These facts served as an early hint that gravitational waves exist, and that they travel at $c$ .\\

One can utilize the remaining coordinate freedom (since the Lorenz gauge is only a \textit{partial} gauge fixing) to obtain illuminating expressions for the $C_\munu$ in the so-called \textit{transverse-traceless} gauge\footnote{For a simple way to convert quantities from an arbitrary gauge into the transverse-traceless gauge see \cite[Eq. 6.55-6.57]{carroll}.}, named so because in such a gauge the perturbation $h$ is traceless and thus $h=\bar{h}$. Reusing the labels $x^\mu$ for the coordinate system resulting from the full gauge fixing, one finds \cite[pg.~150]{carroll} that for a wave travelling in the $x^3$ direction\footnote{We use a coordinate system $\{x^0,x^1,x^2,x^3\}$ and label the first coordinate $x^0$ as ``$t$". We choose two of the spatial components of $k$ to be 0, and since the timelike component of $k$ is denoted as the frequency $\omega$, the null condition implies $(k^\mu) = (\omega, 0, 0, k^3) = (\omega, 0, 0, \pm \omega)$ for a future-directed wave vector.} the coefficients $C_\munu$  take the particularly simple form 
\begin{equation}\label{eq:polarisationCoeffs}
(C_\munu)=
\begin{pNiceMatrix}
        0       & 0         &   0                   &  0        \\
        0       & C_{+}         &   C_{\times}                   &  0        \\
        0       & C_{\times}         &   -C_{+}                   &  0        \\
        0       & 0         &   0                   &  0        

\end{pNiceMatrix}
\end{equation}

where the subscripts on the components are justified after computing the effect of such a perturbation on a ring of test particles, and noting that for only $C_+$ nonzero one finds the ring oscillates in a ``$+$" pattern, and for only $C_\times$ nonzero the ring oscillates in a ``$\times$" pattern \cite{Flanagan_2005,azadeh}. The same structure will be observed when we make the transition to the nonlinear theory and attempt to define an analogous ``plane wave" (Sec. \ref{sec:compare}). Note that our perturbation is now fully described by two functions $C_+$ and $C_\times$, suggesting that there exist two linearly independent polarisation states of gravitational radiation.\\

To convince oneself of the physicality of these results, one needs to examine the motion of test particles in such a spacetime. One finds that for non-relativistic test particles, the geodesic equations are solved by a particle whose coordinate location remains constant. In fact the coordinate system can be thought of as ``moving with" the particle, effectively hiding the dynamics from the perspective of our coordinates \cite[Sec. 1.4]{azadeh}. Instead, upon examining the \textit{relative} motion of test particles via the geodesic deviation equation, one finds a periodic oscillation of the test particles, supporting the physicality of such a wave in the weak-field regime.

%% file: Defining Waves in GR/defining.tex
\subsection{Wavelike Exact Solutions}
We now ask ourselves the natural question ``does the full nonlinear theory also admit wavelike solutions?". Furthermore, we wonder if such solutions reduce to those of the linear theory in the weak-field regime. In order to generalize the wave objects of the linearised theory, let us examine which of their properties are covariantly defined (that is, in a coordinate-independent manner). One easily recognizable covariant property is that the ``wave vector" $k$ should be null
\begin{equation}
    g(k,k) = k^\mu k_\mu = 0.
\end{equation}
Further scrutiny of the results of the previous section yields that $k^\mu$ is also an eigenvector of the Riemann tensor with eigenvalue 0, that is
\begin{equation}\label{eq:eigenvector}
    R_{\mu\nu\sigma\rho}k^\rho = 0
\end{equation}
for all $\mu,\nu,\sigma$. One could use these two properties as a starting point for a definition of a wave in general relativity, that is a Lorentzian manifold $(M,g)$ admitting a null vector field\footnote{We change notation from $k$ to $Z$, which is consistent with the notation of \cite{blau} and \cite{Blau_2003}, however many different symbols are used in the literature, such as $V$ \cite{compact,Globke_2016}, $l$ \cite{sippel} and indeed $k$ \cite{stephani_relativity,Bicak,Bicak1999}. 
} 
$Z$ which is an eigenvector of the Riemann tensor with eigenvalue 0. In fact such a spacetime \textit{does} exhibit wavelike behaviour \cite[Ch. 32.3 \& 34.1]{stephani_relativity}, but is rather cumbersome to work with, and is missing some characteristics of the waves in the linearised theory.\\

One such characteristic is as follows: When making the plane wave ansatz Eq. \ref{eq:ansatz}, we assumed the vector field $k$ to be \textbf{constant}. As a result, the rays of the corresponding wave were \textit{parallel} (in the usual Euclidean sense). In order to obtain the same qualitative behaviour, we should not demand that $Z$ be an eigenvector of the Riemann tensor with eigenvalue 0, but rather the stronger condition that $Z$ be \textit{covariantly constant} (which in some sense generalises the notion of ``constant") which is written $\nabla Z = 0$ for $\nabla$ the Levi-Civita connection of the geometry in question.
With this, we attempt the following covariant definition:
\begin{Definition}\label{def:parallelRays} \textup{Parallel Wave.}\\ 
    A parallel wave (wave with parallel rays) is a Lorentzian manifold $(M,g)$ which admits a global, covariantly constant, null vector field $Z$. 
\end{Definition}

The ``rays" of such a wave are the integral curves of the defining vector field $Z$, which are automatically (null) geodesics since $Z$ is covariantly constant. It is justified that we may call such objects ``rays" by the fact that null geodesics correspond to the paths of light rays.\\

\begin{Remark}
    If we had demanded that $Z$ was an eigenvector of the Riemann tensor with eigenvalue zero instead of being covariantly constant, we would obtain an example from a general class of solutions called the ``Degenerate gravitational fields" which contains the pp-waves as a subset (that is, $Z$ being covariantly constant implies that $Z$ is an eigenvector of the Reimann tensor with eigenvalue 0, but the converse is not true). These degenerate vacuum solutions are defined by the property that they admit (at least) one shear-free, geodesic null congruence. For details of this class and in particular the above mentioned example, see \cite[Ch. 32.3 \& 34.1]{stephani_relativity}. The family of geometries admitting at least one shear-free, twist-free, geodesic null congruence splits into the Kundt class (for a non-expanding congruence) and the Robinson-Trautman class (for an expanding congruence). For details of these classes see \cite{griffiths_podolský_2009}, but in this article we focus primarily on the pp-waves.
\end{Remark}

However, another feature of the plane waves in the linearised theory which we have not yet imposed is the \textit{planar} character. A plane wave has a planar wavefront (roughly, the spacelike codimension-2 hypersurface orthogonal to the wave vector), but in general these parallel waves can have curved wavefronts. Although to obtain wavelike behaviour it is not necessary to demand the wavefront be flat (and in fact we will reintroduce this curvature in Sec. \ref{sec:Nhp}), it is standard in the field to make this restriction. This is likely because when considering a curved wavefront, the geometric properties of the wavefront can ``obscure" those fundamental properties of the wave, such as the vanishing of the scalar curvature invariants (Sec. \ref{sec:invariants}). To demand the wavefront is flat, let us define precisely\footnote{In order to define a true ``direction of motion" of the wave and its wavefront, one must specify an observer (or really family of observers). For the details of how the observer can be used to define the wavefront in coordinates, see \cite[Eq. 4.2.1]{exactsolEK}.} the wavefront of a wave:

\begin{Definition}\label{def:wavefront} \textup{The Wavefront of a Parallel Wave.}\\
    If a parallel wave $(M,g)$ is defined by a covariantly constant, null vector field $Z$ (analogous to the ``wave vector" of a plane wave in the linear theory) then the wavefront of such a wave is defined as 
    \[
        Z_\perp / Z,
    \]
    where $Z_\perp \vcentcolon= \{X \in TM ~|~ g(X,Z) = 0\}$ and the quotient is defined by the equivalence relation $X\sim Y \iff Y = X + fZ$ for some smooth function $f$. 
\end{Definition}
We must quotient with the wave vector itself since $Z$ is null, thus $Z \in Z_\perp$ and the natural analogy to electromagnetism suggests that $Z$ itself should not be considered as part of the wavefront. This definition appears in \cite{compact} under the name ``screen bundle", where it is treated rigorously in the context of compact pp-waves. As the authors note, the ``wave" interpretation becomes less clear in the compact case. As will be discussed in Sec. \ref{sec:radiation}, the presence of radiation is characterized by the null \textit{asymptotics} of the spacetime, but a compact manifold does not admit the same notion of ``null infinity" as will be used to define the presence of radiation. With this in mind, we maintain the name ``wavefront" for simplicity. For details of the induced metric on the wavefront see Sec. \ref{sec:ppWavefront}.\\

If we wish to demand that the wavefront be flat, then this is most succinctly described (see \cite{Globke_2016}) by considering the Riemann tensor as a map on bivectors (antisymmetric 2-tensors) in $Z_\perp \wedge Z_\perp$, in which case the flatness condition for the wavefront becomes
\begin{equation}\label{eq:flatWavefront}
    R\rvert_{Z_\perp \wedge Z_\perp} = 0.
\end{equation}

With this, we arrive at the definition of the \textbf{p}lane-fronted waves with \textbf{p}arallel rays (pp-waves).
\begin{Definition}\label{def:pp} \textup{Plane-fronted Wave with Parallel Rays (pp-Wave).}\\
    A pp-wave is Lorentzian manifold $(M,g)$ which admits a global, covariantly constant, null vector field $Z$, in which the curvature tensor satisfies $R\rvert_{Z_\perp \wedge Z_\perp} = 0$.
\end{Definition}

Note that in the literature (for example \cite[Eq. 24.39]{exactSolnsKramer}) a pp-wave is often defined as a Lorentzian manifold admitting a covariantly constant, null vector field (that is, our definition of a \textit{parallel wave}), where it is understood that the name refers to no actual \textit{planar} character. Other works however also include also the curvature condition Eq. \ref{eq:flatWavefront} as is done here, eg. \cite{compact,Globke_2016,sippel}.

\subsubsection{Comparison with the Linearised Theory}\label{sec:compare}
We now set about comparing the features of these pp-waves with those of the waves found in the linear regime. Consider the metric of Minkowski space written in the so-called ``light-cone" coordinates\footnote{Such coordinates are usually written with ``$-2\mathrm{~d} u \mathrm{~d} v$" rather than the positive term in our metric. Simply transforming $v \longrightarrow -v$ yields our description.} 
\begin{equation}\label{eq:lightCone}
    \eta = 2dudv + dx^{2}+dy^{2},
\end{equation}
where the coordinates $u$ and $v$ are defined in terms of the standard $t,x,y,z$ coordinates as 
\begin{equation}
    u \vcentcolon= \frac{z-ct}{\sqrt{2}} ~~~~~~~~~~~~~~~~~~v \vcentcolon= \frac{z+ct}{\sqrt{2}}
\end{equation}
and where we briefly reintroduce the speed of light $c$ for transparency. As we will prove in Sec. \ref{sec:generalCoords}, a 4-dimensional pp-wave metric can locally be written as
\begin{equation}\label{eq:standardPpLin}
    g = 2dudv + H(u, x, y) du^{2} + dx^{2}+dy^{2},
\end{equation}
where the so-called ``characteristic function" $H$ is independent of the coordinate $v$, and where we have suggestively used the same coordinate labels as for the above flat metric. Here $H$ describes the wave (deviation from flat space) in the sense that when $H = 0$, we simply have the above flat metric Eq. \ref{eq:lightCone}. Note that this metric is a solution of the vacuum Einstein equations if and only if $H$ is harmonic in $(x,y)$, that is $(\partial_x^2 + \partial_y^2)H(u,x,y) = 0$. Here we already see a hint of wavelike behaviour. Treating $H$ as a perturbation on the Minkowski background (and thus inheriting the coordinate system of Eq. \ref{eq:lightCone}), we see that the perturbation depends on time only through the coordinate $u$, that is a time-dependence proportional to $z - ct$, as one would expect for a travelling wave.\\

Surprisingly, as in \cite[Above Eq. 29.46]{stephani_relativity}, one can show that the pp-wave metric  Eq. \ref{eq:standardPpLin} in fact \textit{solves the linearised field equations}. This is because even in the general theory, no expressions of quadratic order or higher in $H$ nor its derivatives appear in the field equations for such a spacetime. The primary difference with the linear theory is that $H$ need not be ``small". In this way, we see that the ``standard pp-waves" do in fact generalise the results of the linearised theory. \\

We will find in Sec. \ref{sec:plane} that the simplest pp-wave occurs when the characteristic function $H(u,x,y)$ is \textit{quadratic} in $(x,y)$ (with arbitrary $u$-dependence). Such a pp-wave is typically referred to as a ``plane wave". These spacetimes exhibit the same polarization states as those which can be derived in the linear regime, and this is one reason they are given the name ``plane waves" (shown in Sec. \ref{sec:plane}). For a detailed description of the ``planeness" of such spacetimes, see \cite[Sec. 3]{bondi}. In order to directly compare these simple pp-waves to the plane wave solutions in the linear regime, as in \cite{stephani_relativity}, one ``linearises" the exact solution by assuming the amplitude of the wave is small. The reasoning of Stephani \cite{stephani_relativity} is as follows:
\begin{itemize}
    \item The \textit{vacuum} plane wave metric of linearised gravity can be written as 
    \[
        g= 2dudv + \left(1+f(u)\right) \mathrm{d} x^{2} +\left(1-f(u)\right) \mathrm{d} y^{2},
    \]
    where $f(u)=A \cos (\frac{\omega}{c} (u+\varphi))$ for some frequency $\omega$, phase $\varphi$ and constant $A$. As usual on a Minkowski background, we interpret $u$ as $z-ct$.
    \item The linearised version of the vacuum plane wave metric (pp-wave with $H$ harmonic and quadratic in $(x,y)$) can be written
    \[
        g=2dudv + \left(1+\alpha(u)\right) \mathrm{d} x^{2} +\left(1-\alpha(u)\right) \mathrm{d} y^{2}
    \]
    with the $u$-dependence of $\alpha$ arbitrary, and $\alpha \ll 1$.
    \item The frequency $\omega$ of the linearised theory is fixed by the plane wave ansatz, but the profile functions $\alpha(u)$ of the second case have no predetermined frequency. Therefore the $\alpha(u)$ can be chosen for example as
    \[
        \sum_j A_j\cos \left(\frac{\omega_j}{c} (u+\varphi_j)\right)
    \]
    for small constants $A_j$, which corresponds to a superposition\footnote{For any $\alpha$ one may examine it's Fourier decomposition to obtain such an interpretation.} of waves of varying frequency. In this way, the exact solution plane waves are interpreted as a packet of plane waves of differing frequencies.
\end{itemize}

There is a more convincing reason why one would call such a pp-wave a ``plane wave" based on the algebraic and geometric symmetries of the spacetime, and we will discuss this in the following section.

\subsection{Spacetimes Containing Gravitational Radiation}\label{sec:radiation}
Let us now review two paths by which one can obtain definitions of the presence of wavelike behaviour/radiation in a spacetime, and the ways in which these approaches coincide with our existing definition of an exact solution describing only a wave.
\subsubsection{Algebraic Classification of the Weyl Tensor}
Felix Pirani and Hermann Bondi (independently) pioneered an attempt at defining gravitational waves as exact solutions of the Einstein field equations, using geometric and algebraic principles developed first by Petrov. Our presentation will follow closely that of \cite[pg. 8,9]{azadeh}. The key concept in this endeavour is the Weyl tensor, which is the trace-free part of the Riemann tensor. As such, the Riemann tensor reduces to the Weyl tensor in vacuum regions, where the Ricci tensor (the trace of the Riemann tensor) vanishes. 
\begin{equation}
    R_\munu = 0 \iff C_{\mu\nu\sigma\rho} = R_{\mu\nu\sigma\rho}
\end{equation}
for all $\mu,\nu$, where it is understood that a $C$ with four indices is the Weyl tensor, not to be confused with the (0,2)-tensor $C$ in the plane wave ansatz Eq. \ref{eq:ansatz} of the linearised theory. When looking in particular for gravitational waves (i.e. in vacuum), it is apparent that the relevant object for describing the wave is the Weyl tensor.\\

Pirani's intuition was that for gravitational waves, the Weyl tensor should exhibit special symmetries. The Weyl tensor of a spacetime $(M,g)$ is \textit{conformally invariant}, that is, it is invariant under conformal transformations of the metric:
\begin{align}
g_{\mu \nu} & \longrightarrow g_{\mu \nu}^{\prime}=\lambda^{2} g_{\mu \nu} \\
C_{\mu \nu \sigma}{}^{\rho} & \longrightarrow C^{\prime}_{\mu \nu \sigma}{}^{\rho} =C_{\mu \nu \sigma}{}^{\rho}
\end{align}
for some conformal factor $\lambda: M \mapsto \reals$. Intuitively, the Weyl tensor expresses the tidal forces that a free-falling body feels along a geodesic (see \cite{azadeh}). That the Weyl tensor describes tidal forces (roughly, the relative acceleration felt by two test masses separated by an infinitesimal distance) should sound familiar, as this was how we detected the physical effect of gravitational waves in the linearised theory. It should not be surprising then that the Weyl tensor is the object describing radiation in general relativity. The correspondence between tidal forces and exact gravitational waves has been the subject of much study (often from the perspective of the geodesic \textit{deviation} equation), details of which can be found in the following articles: \cite{tidal,Felice,compass,Bicak1999,Podolsky2012,Podolsky2013}.\\

In 1954, Petrov devised a classification of the algebraic symmetries of the Weyl tensor at each point in a 4-dimensional spacetime, and Pirani independently derived the same classification in 1957. They noted that the Weyl tensor preserves the antisymmetry of antisymmetric 2-tensors (or ``bivectors"), that is for $X_\munu = -X_{\nu\mu}$,
\begin{equation}
    X_\munu C^\munu{}_{\sigma\rho} = Y_{\sigma \rho}
\end{equation}
where $Y_{\munu}$ is also a bivector. By finding the eigenbivectors $X_\munu$ of the Weyl tensor, i.e. bivectors satisfying $X_\munu C^\munu{}_{\sigma\rho} = 2\lambda X_{\sigma \rho}$, one can classify 6 types of algebraic symmetry. The eigenbivectors for a given point $p$ in a spacetime are related to a set of null vectors in $T_p M$ called the ``principal null directions" (PNDs) at $p$, but the specifics of this correspondence are rather complicated. For details see for example \cite{exactSolnsKramer} or \cite[Sec. 7.2-7.4]{petrovBook}. \\

One may wonder why there are 6 symmetry types, but this is simply because the Weyl tensor can have at most 4 linearly independent eigenbivectors, and so the options are:

\begin{equation*}
    \text{Type I:}~~\uparrow\rightarrow\nwarrow\nearrow  ~~~~\text{Type II:}~~\uparrow\uparrow\nearrow\searrow ~~~~ \text{Type D:}~~ \uparrow\uparrow\rightarrow\rightarrow
\end{equation*}
\begin{equation*}
    \text{Type III:}~~\uparrow\uparrow\uparrow\rightarrow ~~~~~  \text{Type N:} ~~\uparrow\uparrow\uparrow\uparrow ~~~~~~  \text{Type O:}~~ C_{\mu\nu\sigma\rho} = 0
\end{equation*}
where aligned arrows represent linearly dependent PNDs. The \textit{Bel criteria} are the conditions on the Weyl tensor $C_{\mu\nu\sigma\rho}$ (in a special coordinate system) such that it is of one of the above types. The Bel criterion for a type N spacetime is that the metric admits a null vector field $k^\rho$
\begin{equation}
    C_{\mu\nu\sigma\rho} k^\rho = 0 
\end{equation}
This condition should again look very familiar, as it was one of the two covariantly defined properties of the \textit{wave vector} $k$ in the linear theory, where the Riemann tensor is replaced by only the Weyl tensor (which it indeed reduces to in a vacuum region). The four coinciding PNDs indeed correspond to the wave vector of the linear theory, but \textbf{also} to the covariantly constant, null vector field $Z$ in the definition of a pp-wave Eq. \ref{def:pp}. By this we mean that the pp-wave spacetime is everywhere algebraically special, and is of Petrov type N.\\

In this way, the Petrov type N represents the presence of wavelike behaviour in a spacetime. Note that the Petrov type can vary from region to region in a spacetime (though not all ``transitions" are possible, see \cite{azadeh}), and so the Weyl tensor of what we could reasonably consider a radiative spacetime should be of type N in the far-field (towards null infinity). Such a statement is made precise by the ``peeling theorem" \cite{wald, penrosePeeling,gerochPeeling}, which describes the asymptotic behaviour of the Weyl tensor as one approaches null infinity. For $r$ an affine parameter along a null geodesic $\gamma$ from a point $p$ to null infinity, as $r\rightarrow \infty$, the Weyl tensor can be written in a parallelly propagated frame along $\gamma$ as
\begin{equation}
    C_{\mu\nu\sigma\rho} = \frac{C_{\mu\nu\sigma\rho}^{\text{(N)}}}{r} + \frac{C_{\mu\nu\sigma\rho}^{\text{(III)}}}{r^2} + \frac{C_{\mu\nu\sigma\rho}^{\text{(II)}}}{r^3} + \frac{C_{\mu\nu\sigma\rho}^{\text{(I)}}}{r^4} + \dots
\end{equation}
where the superscript on each term on the right hand side represents the Petrov type of that tensor. Roughly\footnote{For the subtleties in such an interpretation we direct the reader to the afformentioned references \cite{penrosePeeling} and \cite{gerochPeeling}.}, towards null infinity one finds that the dominant behaviour comes from the type N component. This expansion bears a striking resemblance to the multipole expansion of the electromagnetic potentials, wherein again only the $\sim1/r$ term contributes to radiation.

\begin{Remark}\label{remark:PenroseAsymptotic}
    We pause to mention here the more geometric notion of asymptotic behavior at infinity due to Penrose \cite{Penrose-APFSP}, where infinity is regarded as a three-dimensional boundary corresponding to $\Omega = 0$ in the definition of the following conformal metric
    $$
    g = \Omega^2 \tilde{g},
    $$
    where $\tilde{g}$ is the original spacetime metric.  The key is that one can treat infinity as a 3-dimensional boundary while still studying those physical properties of the original spacetime metric $\tilde{g}$ that are conformally invariant.  For a comprehensive treatment of this notion of conformal infinity, consult \cite{Penrose-SSP}; for its more recent use in holography and the AdS-CFT correspondence, consult, e.g., \cite{BN-lightcone}.
\end{Remark}

\begin{Remark}
    It is worth now stating precisely what one means by gravitational \textit{radiation}. As in \cite{tidal}, gravitational radiation is the transfer of energy via gravitational waves to null infinity, that is gravitational radiation is present in the asymptotic regime of an isolated dynamical system in GR such as that in the Christodoulou-Klainerman spacetimes \cite{Christodoulou:1993uv}.
\end{Remark}

In 1957, Pirani attempted to define the presence of gravitational radiation as being modelled by a spacetime which was \textit{everywhere} algebraically special with certain type \cite{pirani1957}, but eventually published new work with Robertson and Bondi \cite[Sec. 4]{bondi} in which they claimed that such a definition was too restrictive and in fact only applies to pure radiation; it would not describe the radiation from a system of charges (gravitational or electromagnetic) at a finite distance. As such, they revised the definition of a spacetime containing gravitational radiation to a spacetime which is \textit{asymptotically} type N. One reason for this is that a plane wave is everywhere\footnote{Note that the ``sandwich waves" and ``impulsive waves" mentioned in Table \ref{tab:definitions} are in fact everywhere type \textit{O} (flat) except for a curved region in which they are type N. Additionally, for certain impulsive waves such as the Aichelburg-Sexl solution \cite{aichelburg}, the geometry is also asymptotically flat in the transverse directions ($x$ and $y$ here).}
 type N (again in the original classification of Petrov), and in the far-field, gravitational radiation should approximate the plane wave. The everywhere type $N$ spacetimes contain the ``pp-waves" defined above as a subclass, see \cite[Sec. 18.2]{griffiths_podolský_2009}.

\subsubsection{Groups of Motions (Symmetry)}
In an attempt at a purely geometric definition of gravitational waves, Bondi, Pirani and Robinson began by attempting to define covariantly the \textit{plane wave}. They do this by demanding that the gravitational plane wave of general relativity should ``possess an analogous degree of symmetry to that possessed by 
plane electromagnetic waves in flat space-time" \cite[Sec. 2]{bondi}. As mentioned in the original paper, this approach ensures that one avoids the so-called ``coordinate waves" which are apparent wavelike behaviours which are removed by a diffeomorphism (and thus, simply artifacts of the coordinates chosen).\\

\noindent Consider a plane wave in Minkowski space with wave vector in the positive $z$ direction\footnote{We use the standard coordinate system $\{t,x,y,z\}$.}. There is one clear symmetry of such a wave, and that is the planar wavefront. More precisely, translations in the $x$ and $y$ directions leave our description invariant. Another symmetry is due to the translation of the wavefronts themselves, i.e. the translation along the null 3-surfaces $z - t = \text{const}$ in units where $c = 1$. In fact, there are an additional 2 less obvious symmetries known as the ``null rotations", which are more difficult to see and visualise as their nature is inherently 4-dimensional. In total, we say there exists a 5-parameter group of motions (isometries) under which the plane wave is invariant. The corresponding Killing vector fields for these isometries are given explicitly in \cite[Table 24.5]{exactSolnsKramer} and \cite[Sec. 17.5]{griffiths_podolský_2009}. Using this as inspiration, the authors defined a gravitational plane wave as follows, where ``equivalent" is in reference to a spacetime with metric Eq. \ref{eq:standardPpLin} such that $H$ is quadratic in $(x,y)$ as was briefly mentioned in Sec. \ref{sec:compare}, and is made more explicit in Sec. \ref{sec:plane}. \\

\begin{Definition}\textup{Equivalent Definition: Plane Wave}\label{def:planeSymmetry}\\
    A plane wave is a 4-dimensional non-flat Lorentzian manifold $(M,g)$ which admits a 5-parameter group of isometries.
\end{Definition}
Note that in the original \cite{bondi}, the definition also involves ``Ricci-flat", but this would only correspond to the purely \textit{gravitational} plane waves. The other definitions of the plane wave presented here (via quadratic $H$ in Brinkmann coordinates and via the curvature condition of Definition \ref{def:planecurvature}) include also \textit{electromagnetic} plane wave components in general. Also note that we make no assumption about the structure of the symmetry group; in particular, we do not assume it to have the same group structure as that of a plane wave in electromagnetism. Remarkably, such a property appears as a consequence of our existing assumptions. Such symmetries can be viewed as generated by vector fields, and the explicit form of these generators is given in \cite[Eq. 2.12]{bondi}, for a wave constructed in such a way that it has a finite wave profile\footnote{Such waves have been named ``sandwich waves" since they exhibit a non-flat region (the wave packet) sandwiched between flat regions. Note also that in the limit of shrinking support of the curved region, one obtains the so-called ``impulsive waves'' \cite{podolskyImpulsive}.}. Note also that the gravitational plane wave of Definition \ref{def:planeSymmetry} above is in fact a special case of our pp-wave spacetimes (Definition \ref{def:pp}), and corresponds to the ``plane wave" mentioned in the comparison to the linear theory. These plane waves are described fully in Sec. \ref{sec:classicalPp}.\\

We can also define the plane wave in a covariant manner as in \cite{Globke_2016} as follows, where a ``classical pp-wave" is simply a 4-dimensional pp-wave with planar wavefront (see Sec.~\ref{sec:classicalPp}):

\begin{Definition} \textup{Equivalent Definition: Plane Wave}\label{def:planecurvature}\\
    A plane wave is a classical pp-wave defined via a covariantly constant, null vector field $Z$ which additionally satisfies
    \[
        \nabla_X R = 0 \quad \forall \quad X \in Z_{\perp},
    \]
    where $R$ is the curvature tensor and $Z_\perp \vcentcolon=\{X \in TM ~|~ g(X,Z) = 0\}$.
\end{Definition}

We prove the correspondence of such a definition with the other definitions of a plane wave in Sec. \ref{sec:plane}. For a full discussion of the properties of such waves, the fact that such a definition actually coincides with the algebraic definition of plane waves and the conceptual difficulties involved (e.g. ``to whom is such a gravitational plane wave \textit{planar}?"), see \cite{bondi}. For a succinct overview of the connection between the Petrov classification and the definition of the plane wave in terms of its symmetry group, see \cite[pg.~688]{summary}.\\

Note that all our definitions involve at least \textit{one} lightlike group of motions (symmetry), corresponding to the ``propagation" of the wave. 
There are conditions one may place on a wave such that the wavefront itself is of finite extent (which amount to conditions on the characteristic function $H$ in standard coordinates) and such conditions have relevance to determining the causal character of the wave, as we will see in Sec. \ref{sec:causal}. For a detailed table describing various special cases of gravitational \textit{pp-waves} and their symmetry properties/Killing vector fields, see \cite[pg. 79]{exactsolEK}. \\

The next step in defining the presence of radiation in a spacetime was provided by Trautman, by imposing boundary conditions at infinity in analogy to the Sommerfeld radiation conditions. He showed that in electromagnetism, his conditions restricted one to those solutions of Maxwell's equations with \textit{outgoing} radiative fields. Note that as in the case of the Petrov classification, it is the \textit{asymptotic} behaviour which is used to define the presence of waves. For a review of Trautman's definition in the context of the development of gravitational wave theory, see \cite{summary}, and for Penrose's contribution to the study of asymptotics and their relation to outgoing radiation, see \cite{Penrose-APFSP}.

%% file: coordDesc.tex
\section{The Coordinate Description}\label{sec:coordDesc}
We have defined a parallel wave as a Lorentzian manifold admitting a covariantly constant, null vector field, and a pp-wave as a parallel wave with flat wavefront. In this section, we first derive the most general form of a Lorentzian metric satisfying these conditions, and then discuss the various simplifications which have been studied in the literature. These simplifications remain exact wavelike solutions to the Einstein equations, but have the benefit of being easier to understand and work with. The simplest and most widely known example we call the ``classical pp-wave", which is discussed in Section \ref{sec:classicalPp}.\\

\textbf{Notation:} Our goal is to develop a local coordinate system on a parallel wave of dimension $n$ which we will denote $\{u,v,\mathbf{x}\}$, where $\mathbf{x} = x^1, \ldots ,x^{n-2}$ are the so-called ``wavefront coordinates". This name is justified by examining the definition of a wavefront (Def. \ref{def:wavefront}) in the context of the coordinate description of a parallel wave metric Eq. \ref{eq:generalPW}. We will use Greek indices when referring to all coordinates $\{u,v,\mathbf{x}\}$, and Latin indices (other than the letters $u$ and $v$) when referring to only the wavefront coordinates. For example, the sum $g_{va}X^{a}$ for some vector field $X \in \mathfrak{X}(M)$ (the space of vector fields on $M$) will have $n-2$ terms ($a \neq u,v$), whereas the sum $g_{v\sigma}X^{\sigma}$ will have $n$ terms. To avoid confusion with the coordinates $u$ and $v$, we will not use the typical $\mu$ and $\nu$ Greek indices in this section, and instead we will favor $\sigma, \rho, \gamma$. For the Latin indices, we use $a,b,c$ and $i,j,k$. Additionally, when a coordinate is labelled $x^i$, we will denote its corresponding coordinate vector field by $\partial_{x^i} =\vcentcolon \partial_i$.\\

\subsection{General Parallel Waves and pp-Waves}\label{sec:generalCoords}
Consider the $n$-dimensional Lorentzian manifold $(M,g)$. Denote the covariantly constant null vector field on $M$ by $Z$, that is $\nabla Z = 0$  and $g(Z,Z) = 0$ for $\nabla$ the Levi-Civita connection on $(M,g)$ and $Z$ nontrivial.

\begin{theorem}\label{thm:adapted} \textup{Coordinates adapted to covariantly constant\footnote{Note that for this particular result, one may relax the condition that $Z$ be covariantly constant. For details see Sec. \ref{sec:penroseLimit}. In this context however, $Z$ is always assumed to be covariantly constant.}, null vector field.}\\
If a Lorentzian manifold $(M,g)$ admits a covariantly constant, null vector field $Z$, then in a neighbourhood $U$ of each $p \in M$ there exists a local coordinate chart $\varphi = \{u,v,\mathbf{x}\}$ on $U$ which is ``adapted to $Z$" such that 
 $$Z\rvert_U = \partial_v = \nabla u.$$
\end{theorem}

\begin{proof}
    The proof can be found in Appendix \ref{app:proof}
\end{proof}

The following proposition outlines the properties of the metric $g$ when it is written in these adapted coordinates.

\begin{proposition}\label{prop:metricProperties}
If a Lorentzian manifold $(M,g)$ admits a covariantly constant, null vector field $Z$, and $\{u,v,\mathbf{x}\}$ are the local coordinates adapted to $Z$ of theorem \ref{thm:adapted}, then the metric components in this coordinate system have the following properties on the domain of definition of the coordinates:
\begin{enumerate}[i), nosep]
    \item All metric components are independent of $v$, that is $\partial_v(g_{\mu \nu}) = 0$
    \item $g_{v\sigma} = \delta_{\sigma}^{u}$
    \item $(g_{ab})$ forms a positive-definite matrix, and therefore the embedded codimension-2 submanifolds defined by $u = \text{const}$, $v = \text{const}$ are \textit{Riemannian manifolds}. 
\end{enumerate}
\end{proposition}

\begin{proof}
    $ $\\
    \begin{enumerate}[i), nosep]
    \item A covariantly constant vector field $Z$ is in particular a Killing vector field. By definition of a Killing vector field we have $\mathcal{L}_Z (g) = 0$, but since $Z = \partial_v$ we have $0 = \left[\mathcal{L}_Z (g)\right]_{\sigma \rho} = Z (g_{\sigma \rho}) = \partial_v (g_{\sigma \rho})$.
    \item First note that $Z^{\sigma} = \delta_v^{\sigma}$ and therefore $Z_{\sigma} = g_{v\sigma}$. Then since $Z  = \nabla u = du^\sharp$, we have $Z_{\sigma} = du_{\sigma} = \delta_{\sigma}^{u}$. Therefore $g_{v\sigma} = \delta_{\sigma}^{u}$.
    \item First, the hypersurfaces $\Sigma_c \vcentcolon= u^{-1}(c) = \left\{q \in U: \varphi(q)=\left(c, v(q), x^1(q),\dots, x^{n-2}(q)\right)\right\}$ are \textit{null} hypersurfaces since the normal to these surfaces is the null $grad(u) = Z$. Via the previous point, the normal $Z = \partial_v$ is orthogonal to $\partial_i$ for all $i\in\{1,\dots,n-2\}$ \textit{and to itself} and therefore all these coordinate vectors lie in the null hypersurfaces $\Sigma_c$. \\
    
    Via \cite[Lemma 28, p. 142]{oneill} we have that a null hypersurface can contain only one null vector (here, $Z = \partial_v$ itself) and so the remaining coordinate vector fields must be timelike or spacelike. Via point (2) of the same lemma, we have that there are no timelike vectors, and therefore the $\partial_i$ for all $i\in\{1,\dots,n-2\}$ are \textit{spacelike} and thus $g_{ii}>0$ for all $i$, that is $(g_{ab})$ is positive-definite.
\end{enumerate}
\end{proof}

Using the results of Theorem \ref{thm:adapted} and Proposition \ref{prop:metricProperties}, we can now write the explicit form of the metric $g$ in adapted coordinates for a general parallel wave:
\begin{equation}\label{eq:generalPpOld}
g = 2 d u d v+g_{u u}\left(u, \mathbf{x}\right) d u^{2}+ 2g_{a u}\left(u, \mathbf{x}\right) d x^{a} d u+g_{a b}\left(u, \mathbf{x}\right) d x^{a} d x^{b} 
\end{equation}
The functions $g_{u u}\left(u, \mathbf{x}\right)$ and $g_{a u}\left(u, \mathbf{x}\right)$ will be useful for the classification of parallel wave spacetimes, and we will therefore label them $H\left(u, \mathbf{x}\right)$ and $A_a \left(u, \mathbf{x}\right)$ respectively. We then have the metric of a general parallel wave in local adapted coordinates \cite{Brinkmann1925,coley,Podolský_2009,ortaggio},
\begin{equation}\label{eq:generalPW}
\boxed{g = 2 d u d v+H\left(u, \mathbf{x}\right) d u^{2}+2A_a\left(u, \mathbf{x}\right) d x^{a} d u+g_{a b}\left(u, \mathbf{x}\right) d x^{a} d x^{b}.}
\end{equation}
One could also write this metric in matrix notation as
\begin{equation}\label{eq:generalPpMatrix}
g=
\begin{pNiceMatrix}
        H       & 1         &   A_1                     &  \dots   & A_{n-2}        \\
        1       & 0         &   0                       &  \dots   & 0        \\
        A_1     & 0         & \Block{3-3}{(g_{ab})}  
  &           &         \\
        \vdots  & \vdots    &                           &           &          \\
        A_{n-2}     & 0         &                           &           & 
\end{pNiceMatrix}.
\end{equation}
In fact, this result can be viewed as a special case of a more general result by \cite{walker}, which derives this form of a metric admitting a parallel null \textit{plane} rather than a parallel null vector field. Conceptually the generalisation is simple, as a parallel null $r$-plane is pointwise a set of $r$ linearly independent vectors, such that the field of planes (replacing the vector field in the above example) is a parallel null $r$-dimensional section of the tangent bundle $TM$. In this case, the metric takes a form similar to Eq. \ref{eq:generalPpMatrix}, though with some individual elements replaced by matrix blocks.\\

If we then impose the curvature condition Eq. \ref{eq:flatWavefront} to obtain a pp-wave, as demonstrated in \cite[Appendix A]{Globke_2016} one finds the metric of a general pp-wave in local adapted coordinates
\begin{equation}\label{eq:generalPp}
\boxed{g = 2 d u d v+H\left(u, \mathbf{x}\right) d u^{2}+2A_a\left(u, \mathbf{x}\right) d x^{a} d u+\delta_{a b}\left(u, \mathbf{x}\right) d x^{a} d x^{b}.}
\end{equation}
Note that in the context of pp-waves, these coordinates are sometimes referred to as \textit{Brinkmann coordinates} due to their original discovery \cite{Brinkmann1925} in a primarily mathematical context.\\

The properties of this general metric and some of the various special cases are discussed in Section \ref{sec:properties}. The remainder of this section focuses on defining these special cases, which are obtained by making additional assumptions on $H, A_a, g_{ab}$, the topology of the manifold, or the dimension $n$.\\

\begin{Remark} Gauge Freedom:\\
The gauge freedoms of the parallel wave and pp-wave metrics have been studied carefully, for example by \cite[Sec. 24.5]{exactSolnsKramer} in the $n=4$ case, and \cite[Sec. 6.1]{Podolský_2009} in the $n>4$ case. In vacuum regions it is standard to utilize \textit{local} gauge freedoms to eliminate the cross terms $dx^adu$, though in certain cases one can ``lose" some global information about the nature of the wave source in doing so. Both the process of changing the coordinates to eliminate these terms and extensive detail about which global information is lost in performing such a transformation can be found in \cite{gyratonic}, and will be discussed again in Sec. \ref{sec:gyratonic}. Upon eliminating these terms, the metric locally takes the form 
\begin{equation}\label{eq:PpNoA}
g = 2 d u d v+H\left(u, \mathbf{x}\right) d u^{2}+g_{a b}\left(u, \mathbf{x}\right) d x^{a} d x^{b}
\end{equation}
which one can summarise as 
\begin{equation}
g = 2 d u d v+H\left(u, \mathbf{x}\right) d u^{2}+h(u),
\end{equation}
where $h$ is a $u$-dependent family of Riemannian metrics on the codimension-2 hypersurface $u =$ const, $v=$ const. The construction of this form of the metric can be found for the special case of \textit{classical pp-waves} (Eq. \ref{eq:classicalpp} below) in \cite[Theorem 4.1.3]{exactsolEK}.\\

For such a metric, it was shown in \cite{Globke_2016} that the coordinate changes which leave this form Eq. \ref{eq:PpNoA} invariant are
\begin{align}
\begin{split}
    v &\longrightarrow v^\prime = \frac{1}{a}v + f_1(u, \mathbf{x})\\
    u &\longrightarrow u^\prime = au + b\\
    \mathbf{x} & \longrightarrow \mathbf{x}^\prime = \mathbf{f}_2(u, \mathbf{x}),
\end{split}
\end{align}
where $a \neq 0$ and $b$ are constants and $f_1$, $\mathbf{f}_2$ are smooth functions independent of $v$ on the domain of the coordinate chart. In such coordinates, the metric would retain its form 
\begin{equation}
g = 2 d u^{\prime} d v^{\prime} + H^{\prime}\left(u^{\prime}, \mathbf{x}^{\prime}\right) d u^{\prime}{}^{2}+h^{\prime}(u^\prime).
\end{equation}
The authors showed that this fact may be used to transform to so-called \textit{normal} Brinkmann coordinates centred at $p$, in which it holds that $\varphi(p) = 0 \in \reals^n$ where $\varphi$ is the coordinate chart and 
\begin{equation}
    H(u,\mathbf{0}) = 0,~~~\frac{\partial H}{\partial x^i}(u,\mathbf{0}) = 0
\end{equation}
for all $u$ in an interval around $0$.
\end{Remark}

\subsection{Standard pp-Wave}\label{sec:standardPp}
The class of pp-wave most commonly studied in the physics literature has been referred to by \cite[Eq.~2]{compact} as a \textit{standard pp-wave}. The defining characteristics of a standard pp-wave metric when written in the coordinate chart $\{u,v,\mathbf{x}\}$ of Theorem \ref{thm:adapted} are:
\begin{enumerate}[i), nosep]
    \item The coordinates $\{u,v,\mathbf{x}\}$ exist globally
    \item The metric is written with no cross terms $dx^a du$, that is $A_a = 0$ for all $a$.
\end{enumerate}

and thus our metric takes the form
\begin{equation}\label{eq:standardPp}
g = 2 d u d v+H\left(u, \mathbf{x}\right) d u^{2}+\delta _{a b}d x^{a} d x^{b}.
\end{equation}

One can see that in coordinates, the codimension-2 hypersurface defined by $u =$ const, $v =$ const corresponds precisely to the \textit{wavefront} of Definition \ref{def:wavefront}. Unsurprisingly, for this $n$-dimensional standard pp-wave, the wavefront (or ``transverse space") is simply Euclidean $\reals^{n-2}$.
\\

By assuming the coordinates $u$ and $v$ exist globally, we are making assumptions on the properties of the spacetime manifold $M$. Certainly, that $M$ is simply-connected is a sufficient condition for the coordinate $u$ being global (since then the construction involving the Poincar\'e lemma would hold globally) but this is certainly \textbf{not} a necessary condition (for example the $(N,h)$p-waves of Sec. \ref{sec:Nhp} with any non-simply connected $N$ still admit a global $u$). In the case of the $v$ coordinate, one expects that the integral curves of $Z$ should be complete and non-closed\footnote{It appears the analysis of weakest conditions under which such coordinates exist globally is not present in the literature, and remains an open question.}. Typically physical research involving pp-wave spacetimes begins with the assumption of a Lorentzian manifold $(M = \reals^n,g)$ with a metric of the form above.

\subsection{Classical pp-Waves}\label{sec:classicalPp}
These are the pp-waves for which the wavefront is two-dimensional Euclidean space, that is they are standard pp-waves on $\reals ^4$ such that the metric takes the form
\begin{equation}\label{eq:classicalpp}
    g = 2 d u d v+ H(u, x,y) d u^{2}+ dx^2 + dy^2,
\end{equation}
where the usual adapted coordinates on the wavefront $(x^1,x^2)$ have been relabelled\footnote{In some expressions it will be useful to label these coordinates as the usual $x^i$, such that for example $\sum_{i\in\{1,2\}} x^i = x + y$} to $(x,y)$. This metric is the most widely-known and well-studied pp-wave metric, due to its relevance to physics, and its simplicity while still exhibiting the key features of a pp-wave. The most important types of classical waves are the \textit{plane waves}, whose properties will be discussed extensively in Sec. \ref{sec:causal} and Sec. \ref{sec:EKConj}.

\subsubsection{Plane Waves}\label{sec:plane}
A \textit{plane wave} is a classical pp-wave for which the characteristic function $H(u,x,y)$ is quadratic\footnote{Both in the literature and here, ``quadratic" means that $H$ is \textit{purely} quadratic, that is contains only quadratic terms in the variables $x$ and $y$ as in Eq. \ref{eq:plane} (e.g. $H(u,x,y) = f(u)x^2 + 3xy$ as an example without much physical meaning). This is because any linear or constant terms in $H$ can be removed via a coordinate transformation, as noted in \cite[Sec. 2.2]{blau}.} in $(x,y)$, i.e. the metric of Eq. \ref{eq:classicalpp} wherein
\begin{equation}\label{eq:plane}
    H(u,x,y) = \sum_{i, j=1}^{2} h_{i j}(u) x^{i} x^{j}
\end{equation}
for a symmetric $2\times2$ and $u$-dependent matrix $h_{ij}(u)$. The vacuum Einstein equations imply \cite{gravwave} that $h_{ij}$ should be trace-free, which means we can write 
\begin{equation}
    (h_{ij})(u) = 
    \begin{pNiceMatrix}
        f_+(u) & f_\times (u)\\
        f_\times (u) & -f_+ (u)
\end{pNiceMatrix}.
\end{equation}
Had we wanted to describe a purely \textit{electromagnetic} wave rather than a gravitational wave, one should have $(h_{ij}) = \text{diag}(f(u),f(u))$ for some arbitrary smooth $f$. A sandwich wave is obtained when the support of the profile functions is \textit{compact}; for details see \cite[Eq. 2.1]{ppgeometry} and \cite[Sec. 17.4]{griffiths_podolský_2009}. Note that the presence of two functions necessary to describe the wave, as in the linear regime, means that the gravitational wave described by such a metric possesses two linearly independent polarization states. Note that we have used the analogous subscripts as we had on the coefficients $C_{\munu}$, as the $f_+$ and $f_\times$ functions again describe the components of the wave in each polarisation state. If we had not imposed the vacuum condition, the plane wave would instead have described a coupled system of both gravitational and electromagnetic plane waves. Such plane waves were originally studied in \cite{baldwin} and then by \cite{brdicka}.\\

Let us now examine the affect of these polarisation states as in \cite[pg.~94]{gravwave}, where we skip some steps due to the similarity with the analysis of the linear regime. For a plane wave, the geodesic equation for $u$ is simply $\ddot{u} = 0$, that for $v$ is
\begin{equation}
\ddot{v}=\frac{1}{2} \left(f_{+}^{\prime}(u)(x^2 - y^2)+2 f_{\times}^{\prime}(u) xy\right) \dot{u}^{2} +\big(f_{+}(u)(x \dot{x}-y \dot{y})+f_{\times}(u)\left(x \dot{y}+y \dot{x}\right)\big) \dot{u}
\end{equation}
and for $x$ and $y$ we have
\begin{equation}
    \begin{pNiceMatrix}
        \ddot{x} \\
        \ddot{y}
    \end{pNiceMatrix}
    = \frac{1}{2}
    \begin{pNiceMatrix}
        f_+(u) & f_\times (u)\\
        f_\times (u) & -f_+ (u)
    \end{pNiceMatrix}
    \begin{pNiceMatrix}
        x \\
        y
    \end{pNiceMatrix}.
\end{equation}
Since $\ddot{u} = 0$, we have that $u(s) = as + b$ for curve parameter $s$ and $a,b \in \reals$. Therefore as the affine parameterisation along a geodesic is only unique up to a transformation of the form $s \mapsto cs + d$, $u$ itself can be used as an affine parameter and we may take $u(s) = s$.\\

For the ``$+$" mode, we have $f_\times = 0$, and one finds the geodesic equations reduce to 
\begin{equation}
    \begin{pNiceMatrix}
        \ddot{x}(s) \\
        \ddot{y}(s)
    \end{pNiceMatrix}
    = \frac{f_+(s)}{2}
    \begin{pNiceMatrix}
        x(s) \\
        -y(s)
    \end{pNiceMatrix}.
\end{equation}
That is, the motion decouples and takes place only in the transverse directions (as expected by analogy with the linear theory). This motion is such that where $f_+(s)$ is positive, there is a ``focusing" in the $x$ direction and a defocusing in the $y$ direction. Where $f_+$ is negative, one sees the converse effect.\\

By introducing coordinates $(w,z)$ rotated by $45^\circ$ relative to $(x,y)$, and taking the ``$\times$" polarisation mode $f_+ = 0$, one finds precisely the same equation of motion for the rotated variables
\begin{equation}
    \begin{pNiceMatrix}
        \ddot{w}(s) \\
        \ddot{z}(s)
    \end{pNiceMatrix}
    = \frac{f_\times(s)}{2}
    \begin{pNiceMatrix}
        w(s) \\
        -z(s)
    \end{pNiceMatrix},
\end{equation}
where
\[
    \begin{pNiceMatrix}
        w \\
        z
    \end{pNiceMatrix}
    = \frac{1}{\sqrt{2}}
    \begin{pNiceMatrix}
        1 & 1\\
        -1 & 1
    \end{pNiceMatrix}
    \begin{pNiceMatrix}
        x \\
        y
    \end{pNiceMatrix}.
\]
Thus the two polarization modes have precisely the same effect as in the linearised theory, but now there is no requirement that the separations be ``small". This is in line with the interpretation of the characteristic function $H$ as corresponding to the perturbation $h_{\munu}$ of the linear theory, but without the requirement that it be ``small" in some sense.\\

We now demonstrate that the above expression for the metric of a plane wave (Eq. \ref{eq:plane}) corresponds to our previous definitions of a plane wave. The correspondence between the dimension of the symmetry group and the form of the line element has already been succinctly and fully described by \cite[Table, pg 79]{exactsolEK}, and so we will not reproduce the calculation here. This establishes the connection with Definition \ref{def:planeSymmetry}, and we now illustrate the connection with Definition \ref{def:planecurvature}. 
\begin{Lemma}
    The plane wave of Definition \ref{def:planecurvature} corresponds to a classical pp-wave (Eq. \ref{eq:classicalpp}) for which the characteristic function $H$ in Brinkmann coordinates is \textit{quadratic} in $(x,y)$. That is, the condition 
    \[
        \nabla_X R = 0 \quad \forall \quad X \in Z_{\perp},
    \]
    where $Z=\partial_v$ in these coordinates, $R$ is the curvature tensor and $Z_\perp \vcentcolon=\{X \in TM ~|~ g(X,Z) = 0\}$ is equivalent to $H_{xxx}=H_{yxx}=H_{xyy}=H_{yyy}=0$ for classical pp-waves.
\end{Lemma}
\begin{proof}
    First note that $\partial_x$ and $\partial_y$ are elements of $Z_\perp$. Let us begin by examining $\nabla_{\partial_x}R$ which we assume to be 0, and we will see that this implies $H_{xxx}=H_{yxx}=0$. 
    \begin{eqnarray}
    0 &=& (\nabla_{\!\partial_x}R)(\partial_u,\partial_x,\partial_u)\\
    &=& \partial_x(R(\partial_u,\partial_x)\partial_u) - R(\cd{\partial_x}{\partial_u},\partial_x)\partial_u\nonumber - \cancelto{0}{R(\partial_u,\cd{\partial_x}{\partial_x})\partial_u} - R(\partial_u,\partial_x)\cd{\partial_x}{\partial_u}\nonumber\\
    &=&\!\! \partial_x(\nabla_{\!\partial_u}\nabla_{\!\partial_x}\partial_u - \nabla_{\!\partial_x}\nabla_{\!\partial_u}\partial_u) - \frac{H_x}{2}\cancelto{0}{R(\partial_v,\partial_x)\partial_u}\nonumber -\frac{H_x}{2}\cancelto{0}{R(\partial_u,\partial_x)\partial_v}\nonumber\\
    &=&\!\! \cancelto{0}{\frac{H_{xux}}{2}\partial_v - \frac{H_{uxx}}{2}\partial_v} + \frac{H_{xxx}}{2}\partial_x + \frac{H_{yxx}}{2}\partial_y,\nonumber
    \end{eqnarray}
    
    where we have used that the nonzero Christoffel symbols are given by
    \begin{eqnarray}
    \nabla_{\partial_{x}} \partial_{u} &=&\nabla_{\partial_{u}} \partial_{x}=\frac{H_{x}}{2} \partial_{v}, \\
    \nabla_{\partial_{y}} \partial_{u} &=&\nabla_{\partial_{u}} \partial_{y}=\frac{H_{y}}{2} \partial_{v}, \\
    \nabla_{\partial_{u}} \partial_{u} &=&\frac{H_{u}}{2} \partial_{v}-\frac{H_{x}}{2} \partial_{x}-\frac{H_{y}}{2} \partial_{y}.
    \end{eqnarray}
    
    Thus $H_{xxx}=H_{yxx}=0$, and the remainder of the proof then follows by considering $(\nabla_{\!\partial_y}R)(\partial_u,\partial_y,\partial_u)$, from which the result is obtained in precisely the same manner as for $\partial_x$. The reverse direction of the equivalence then follows from the fact that $Z_\perp$ is pointwise spanned by $\partial_x,\partial_y$ and $\partial_v$, and that $\partial_v$ is a Killing vector field.
\end{proof}

\subsection{Gyratonic pp-Waves}\label{sec:gyratonic}
The gyratonic pp-waves are those pp-waves with nonvanishing $A_a$, that is the general metric can be written as
\begin{equation}\label{eq:gyratonicPp}
g = 2 d u d v+H\left(u, \mathbf{x}\right) d u^{2}+2A_a\left(u, \mathbf{x}\right) d x^{a} d u+g_{a b} d x^{a} d x^{b}, 
\end{equation}

but note that the gyratonic pp-waves may also be studied with flat wavefront ($g_{ab} = \delta_{ab}$) as in \cite{gyratonic}. Such pp-waves have been studied extensively, for example in \cite{seminal}, and in \cite{gyratonic}, wherein work by \cite{bonnor} is used to conclude that in the Ricci-flat case, they correspond to the exterior vacuum field of spinning particles moving with the speed of light. In reference to the off-diagonal terms with coefficients $A_a$, the authors state:
\begin{quote}
    \noindent \textit{In vacuum regions it is a standard and common procedure to completely remove these functions by a gauge (coordinate) transformation. However, such a freedom is generally only local and completely ignores the global (topological) properties of the spacetimes. $\dots$ In particular the possible rotational character of the source of the gravitational waves (its internal spin/helicity) is obscured.}
\end{quote}
What one finds (\cite[Sec. 4]{gyratonic}) is that the physical characteristics one can define in a pp-wave spacetime can be obscured via the local gauge transformations which eliminate the $A_a$, and in general it may be necessary to keep such terms. Most notably, one should pay close attention to such terms when attempting to define the angular momentum density of pp-waves in an analogous manner to the linearised theory \cite{MTWGrav}. In the end, such a physical property depends manifestly on the $A_a$ via the contour integral (see \cite[Eq. 33]{gyratonic})
\begin{equation}
    \oint_{C} A_{a} dx^{a},
\end{equation}
where $C$ is a (not completely arbitrary) contour in the transverse space.\\

\subsection{$\mathbf{(N,h)}$p-Waves}\label{sec:Nhp}
These spacetimes are a subclass of the parallel waves which roughly correspond to a standard pp-wave with a Riemannian manifold replacing the planar wavefront of a pp-wave. That is, they are the parallel waves which the following conditions hold:
\begin{enumerate}[i), nosep]
    \item In the adapted coordinates of theorem \ref{thm:adapted}, the metric components of the wavefront $g_{ab}$ are independent of the coordinate $u$.
    \item The spacetime decomposes as $M = \reals^2 \times N$ where $(N,h)$ is a connected Riemannian manifold\footnote{Previously the metric components on the wavefront were written locally as $g_{a b}\left(\mathbf{x}\right) d x^{a} d x^{b}$ but here we use the label $h$ to refer to the global metric on the wavefront.}. Note that this implies the coordinates $u$ and $v$ are globally defined.
\end{enumerate}

This amounts to a general parallel wave metric Eq. \ref{eq:generalPp} with the additional constraint that the metric on the transverse space $h$ be independent of $u$. The name we suggest for such spacetimes is in analogy to the ``pp-wave" spacetimes (\textbf{p}lane-fronted waves with \textbf{p}arallel waves) as here we have a wavefront $(N,h)$ and the rays remain \textbf{p}arallel, as they are the integral curves of $Z$ and $Z$ remains, as always, covariantly constant. Such spacetimes have also been called ``generalised plane waves" \cite{candelaCompleteness} and ``PFWs" (plane-fronted waves) \cite{ppgeometry} \& \cite{Nfronted}, but the authors find this suggested naming scheme to be the most transparent and accurate. We may write the $(N,h)$p-wave metric as
\begin{equation}\label{eq:NPp}
g = 2 d u d v+H\left(u, \mathbf{x}\right) d u^{2}+2A_a\left(u, \mathbf{x}\right) d x^{a} d u+h.
\end{equation}
We can write this metric without referencing coordinates on $N$ if we instead consider $H$ as a map $H\colon \reals \to C^{\infty}(N)$. That is for each $u$, $H$ is a smooth function on $N$. Similarly for the mixed terms $dx^{a}du$, we define $A \colon \reals \to \Gamma(T^*N)$ where $\Gamma(T^*N)$ is the space of sections of the cotangent bundle of $N$. With these redefinitions (unique to this Section) we may write g as
\begin{equation}\label{eq:NPpInvariant}
    g = 2 d u d v+H(u) d u^{2}+2 A(u) du + h.
\end{equation}
Such spacetimes have been studied extensively in \cite{ppgeometry} \& \cite{Nfronted}, in which the geodesic completeness, geodesic connectedness and causality have been determined. 

\subsection{Rosen Coordinates of Plane Waves}\label{sec:rosen}
The coordinates for plane waves which make manifest the symmetries/Killing vector fields are called Rosen coordinates, after \cite{rosen}. The transformation between Brinkmann and Rosen coordinates is well-documented, for example see \cite[Appendix A]{Blau_2003} and \cite[Sec. 2.8]{blau}. We simply present the local form of a plane wave metric in Rosen coordinates, where we use capital letters for the coordinates $U$ and $V$ to emphasize that they are \textbf{not} the same coordinate functions as in Brinkmann coordinates.
\begin{equation}\label{eq:rosen}
    g = 2 d U d V+K_{ij}\left(U\right) dy^idy^j,
\end{equation}
where $K_{ij}$ is positive-definite on the domain of validity of these coordinates. Note that in such coordinates, the Minkowski metric could be represented as
\begin{equation}\label{eq:rosenMinkowski}
    \eta = 2 d U d V+ \delta_{ij}dy^idy^j
\end{equation}
which is simply the usual metric written in light-cone coordinates. Rosen coordinates can often exhibit (coordinate) singularities, and are therefore often avoided in favour of Brinkmann coordinates \cite[Sec. 2.9]{blau}.\\

Generically, the plane wave metric has $2n - 3$ linearly independent Killing vectors, which in a suitable basis generate the Heisenberg algebra \cite[Sec. 2.1]{Blau_2003}. In Rosen coordinates, half ($+1$) of the Killing vector fields are manifest (independent of $K_{ij}$) and the remaining symmetries can be obtained in terms of $K^{ij}$, the inverse of $K_{ij}$. The Killing vector fields are thus (as in \cite[Eq. 2.11]{rosenSymm})
\begin{equation}
    e_{+}=\frac{\partial}{\partial V}, \quad e_{i}=\frac{\partial}{\partial y^{i}}, \quad  e_{i}^{*}=y^{i} \frac{\partial}{\partial v}-\sum_{j} \int K^{i j}(U) d U \frac{\partial}{\partial y^{j}}.
\end{equation}

These correspond to the defining symmetry of the parallel wave $Z$ and the translations and rotations of the $y^j$. Note that the $e_{i}^{*}$ are the usual rotations when we have $K_{ij} = \delta_{ij}$, that is the Minkowski metric eq. \ref{eq:rosenMinkowski}.

%% file: Properties/invariants.tex
\section{Properties}\label{sec:properties}
\subsection{Vanishing Scalar Invariants}\label{sec:invariants}
A well-known property of the pp-wave geometries is that all scalar curvature invariants (a scalar constructed from the metric, Riemann
tensor and covariant derivatives of the Riemann tensor) are zero\footnote{Of course in the context of a general parallel wave, the scalar invariants of the wavefront will be inherited by the full spacetime.} \cite{allVSI, ortaggio}. Here, we will present a proof that all curvature invariants of the \textit{plane waves} vanish, and for the case of the general pp-wave, we direct the reader to \cite{allVSI}. There are two approaches to prove this fact, the first by explicitly calculating the curvature tensor and the second by showing that each point $p$ in a plane wave spacetime is the fixed point of a homothety, and that any curvature invariant must be 0 at such a point. We will present the second such approach here, the proof of which is due to Schmidt \cite{schmidt}, where we follow closely the presentation in \cite{blau}. 
\begin{theorem}
All curvature invariants of a plane wave vanish.
\end{theorem}

\begin{proof}
We will proceed via the following series of arguments:
\begin{enumerate}
    \item An elementary curvature invariant cannot be invariant under constant rescalings of the metric (called a \textit{homothety}).
    \item If there exists a coordinate transformation which induces a homothety, then due to the previous point, at the fixed points of the transformation (i.e. points which are invariant under the transformation) any elementary curvature invariant must be 0.
    \item Any point in a plane wave is the fixed point of a homothety
\end{enumerate}

These statements are proved as follows:
\begin{enumerate}
    \item A general curvature invariant of a manifold $(M,g)$ is constructed from the metric and elementary curvature invariants. An elementary curvature invariant is obtained by taking covariant derivatives of the Riemann tensor
    \[
        \nabla_{\mu_{1}} \ldots \nabla_{\mu_{p}} R_{\nu \lambda \rho}{}^{\mu}
    \]
    and ``tracing out" all free indices with the inverse metric $g^\munu$. The Levi-Civita connection $\nabla$ is invariant under a constant rescaling of the metric (homothety), which is conformal transformation, in which the conformal factor $\lambda$ is a nonzero constant
    \[
        g_\munu \longrightarrow \tilde{g}_\munu = e^{2\lambda}g_\munu
    \]
    That is, we have a second manifold $(M,\tilde{g
    })$ conformally related to $(M,g)$ (the homothety is in particular \textit{not} an isometry). Since $\nabla$ is invariant under such a transformation, so too is the Riemann tensor. Since the (certainly not invariant) inverse metric is required to make a scalar, the elementary invariants cannot be invariant under such a homothety. Rather, a curvature invariant $J$ will change as 
    \[
        J(x) \longrightarrow e^{m\lambda}J(x)
    \]
    for some $x \in M$ and some natural number $m$ which depends on the order of $J$ (number of covariant derivatives). 
    \item Assume there exists a coordinate transformation of the Lorentzian manifold $(M,g)$ which induces a homothety with $x$ a fixed point. Since $x$ is a fixed point of the homothety we have
    \[
        J(x) = e^{m\lambda}J(x)
    \]
    differing from above in the equals sign alone. Such an equality can only hold (for natural $m$ and constant nonzero $\lambda$) if $J(x) = 0$.
    \item We simply need to construct the coordinate change for plane waves which induces a nontrivial homothety. As in Sec. \ref{sec:rosen}, any plane wave metric can be written in the so-called ``Rosen coordinates" as
    \[
        g = 2dudv + g_{ij}(u)dy^idy^j.
    \]
    Such a form exhibits obvious translational symmetry in the $y^j$ and $v$ directions. Due to these symmetries, without loss of generality we can take a general point to be written as $x = (u_0,0,0)$, which is fixed point of the coordinate transformation
    \[
        (u,v,y^j) \longrightarrow (u,\lambda^2 v, \lambda y^j)
    \]
    for some constant $\lambda$. Such a coordinate transformation is in fact a homothety, and scales the metric as $g \longrightarrow \lambda^2 g$. Since we have shown that this is true for general $u_0$, the result holds for any point $(u,v,\mathbf{y})$ of a plane wave.
\end{enumerate}
\end{proof}
For further details of all classes of spacetimes in which the curvature invariants identically vanish, see \cite{allVSI}.

%% file: Properties/wavefront.tex
\subsection{pp-waves via their Wavefronts}\label{sec:ppWavefront}
As mentioned in Definition \ref{def:wavefront} above, a distinguishing feature of a null vector field $Z$ is, of course, that it lies in its own orthogonal complement, $Z_{\perp}$, leading to the Wavefront $Z_{\perp}/Z$, a vector bundle whose elements are equivalence classes ``$[X]$" of vector fields $X$ orthogonal to $Z$.  Because such vector fields are necessarily spacelike (see, e.g., \cite[Lemma~28, p.~142]{oneill}), $Z_{\perp}/Z$ will inherit a (positive-definite) inner product from the Lorentzian metric $g$.  It turns out that when $Z$ is also \emph{parallel}, as it is a for a pp-wave, then $Z_{\perp}/Z$ will also inherit a well defined linear connection, and this can be used to give an alternative\,---\,and very geometric\,---\,definition of a pp-wave.  This alternative formulation of a pp-wave, which we now provide, is well known; see, e.g., \cite{caja}, \cite[Proposition~3]{compact}. In the following, $\Gamma(E)$ represents the space of sections of the vector bundle $E$.
\begin{theorem}\label{thm:QB}
Let $(M,g)$ be a Lorentzian manifold and $Z$ a null, parallel vector field defined in an open subset $\mathscr{U} \subseteq M$, with orthogonal complement $Z_{\perp} \subset T\mathscr{U}$.  Then the wavefront $Z_{\perp}/Z$ admits a positive-definite inner product $\bar{g}$,
\[
\bar{g}([X],[Y]) \defeq g(X,Y)\hspace{.2in}\text{for all}\hspace{.2in}[X], [Y] \in \Gamma(Z_{\perp}/Z),
\]
and a corresponding linear connection $\overline{\nabla}\colon \mathfrak{X}(\mathscr{U})\times \Gamma(Z_{\perp}/Z) \rightarrow \Gamma(Z_{\perp}/Z)$,
$$
\conn{V}{Y} \defeq [\cd{V}{Y}] \hspace{.2in}\text{for all}\hspace{.2in} V \in \mathfrak{X}(\mathscr{U})\hspace{.2in}\text{and}\hspace{.2in}[Y] \in \Gamma(Z_{\perp}/Z).
$$
This connection is flat if and only if $(\mathscr{U},g|_{\mathscr{U}})$ is a pp-wave.
\end{theorem}
\begin{proof}
The metric $\bar{g}$ will be well defined, and positive definite, whenever $Z$ is null; indeed, every $X \in \Gamma(Z_{\perp})$ not proportional to $Z$ is necessarily spacelike, so that $\bar{g}$ is nondegenerate (and positive-definite), and if $[X] = [X']$ and $[Y] =[Y']$, so that  $X' = X +fZ$ and $Y' = Y+kZ$ for some smooth functions $f,k$, then
$$
\bar{g}([X'],[Y']) = g(X',Y') = g(X,Y) =  \bar{g}([X],[Y]).
$$
On the other hand, the connection $\overline{\nabla}$ requires $Z$ to be parallel or else it is not well defined: $\cd{V}{Y}  \in \Gamma(Z_{\perp})$ if and only if $Z$ is parallel, in which case
\[
\conn{V}{Y'} = [\cd{V}{Y'}] = [\cd{V}{Y}]+\cancelto{0}{[V(k)Z]} + \cancelto{0}{[k\cd{V}{Z}]}\, = \conn{V}{Y}.
\]
That $\overline{\nabla}$ is indeed a linear connection follows easily.  Now, if this connection is flat, then by definition its curvature endomorphism, which is the mapping
\begin{equation}
\overline{\text{R}}\colon \mathfrak{X}(\mathscr{U})  \times \mathfrak{X}(\mathscr{U}) \times \Gamma(Z_{\perp}/Z) \rightarrow \Gamma(Z_{\perp}/Z),\nonumber
\end{equation}
whose action is given by
\begin{equation}
\overline{\text{R}}(V,W)[X] \defeq \overline{\nabla}_{\!V}[\overline{\nabla}_{\!W}[X]] - \overline{\nabla}_{\!W}[\overline{\nabla}_{\!V}[X]] - \conn{[V,W]}{X},\nonumber
\end{equation}
will vanish, for any section $[X] \in \Gamma(Z_{\perp}/Z)$ and vector fields $V,W \in \mathfrak{X}(\mathscr{U})$. Using the metric $\bar{g}$, this flatness condition is equivalent to
\[
\bar{g}(\overline{\text{R}}(V,W)[X],[Y]) = 0\hspace{.1in}\text{for all}\hspace{.1in}V,W \in \mathfrak{X}(\mathscr{U})\hspace{.1in},\hspace{.1in}[X],[Y] \in \Gamma(Z_{\perp}/Z).
\]
But if we unpack the definitions of $\overline{\nabla}$ and $\bar{g}$, we see that
\begin{equation}
\label{eqn:RtoR}
\bar{g}(\overline{\text{R}}(V,W)[X],[Y]) = \text{Rm}(V,W,X,Y) = \text{Rm}(X,Y,V,W).
\end{equation}
It follows that $\overline{\text{R}} = 0$ if and only if $R(X,Y)V = 0$ for all $X,Y \in \Gamma(Z_{\perp})$ and $V \in \mathfrak{X}(\mathscr{U})$; by \eqref{eq:flatWavefront} and Definition \ref{def:pp}, this is precisely the condition to be a pp-wave.
\end{proof}

%% file: Properties/PenroseLimit.tex
\subsection{Penrose Limits}\label{sec:penroseLimit}
We now outline the importance and prove the existence of the famous ``Penrose limit", which assigns a plane wave metric Eq. \ref{eq:plane} as a limit of any spacetime $(M,g)$ in a neighbourhood of a null geodesic $\gamma$. This is not a property of the parallel wave metrics, but rather a remarkable feature of all spacetimes. This fact was originally demonstrated by Penrose in 1976 \cite{penroseLimit}, where he described the limiting procedure as a null analogy to the procedure by which one obtains the tangent space (that is, ``zooming in" on a small neighbourhood and scaling those neighbourhoods up in a complementary manner).  It is worth pointing out that applications of Penrose's limit in physics continue to the present day, particularly in higher dimensions and in relation to string theory and the AdS/CFT correspondence; see, e.g., \cite{blau2002penrose,berenstein2002strings,blau2004penrose} and the references therein.
\\

 We adopt a different notation to that of Penrose's work to be consistent with the majority of modern literature regarding pp-waves, and in particular Theorem \ref{thm:adapted} of this article. We take inspiration from the discussion of \cite{phillip2006}, who is consistent in explicitly writing the appropriate pullbacks which appear only implicitly in the original work~\cite{penroseLimit}.

\begin{theorem}\label{thm:nullCoords}
    Consider an
    $n$-dimensional Lorentzian manifold $(M,g)$. In a neighborhood of a point on any conjugate-point free portion $\gamma'$ of a null geodesic $\gamma$, one can write the metric g in the so-called ``null coordinates" as
    \begin{equation}\label{eq:eq:nullCoords}
    g = 2 d u d v+H d u^{2}+2A_a d x^{a} d u+g_{a b} d x^{a} d x^{b},
    \end{equation}
    where $H$, $A_a$ and $g_{ab}$ (with $a,b \in 1,\dots,n-2$) are smooth functions of the coordinates and $(g_{ab})$ is a positive-definite matrix, i.e. a family of Riemannian metrics on the $(n-2)$-dimensional embedded submanifolds defined by $u=\text{const}, v=\text{const}$. One could represent this metric in matrix notation as
    \begin{equation}\label{eq:nullCoordsMatrix}
    g=
    \begin{pNiceMatrix}
        H       & 1         &   A_1                     &  \dots   & A_{n-2}        \\
        1       & 0         &   0                       &  \dots   & 0        \\
        A_1     & 0         & \Block{3-3}{(g_{ab})}  
                                                        &           &         \\
        \vdots  & \vdots    &                           &           &          \\
        A_{n-2}     & 0         &                           &           & 
    \end{pNiceMatrix}.
    \end{equation}
    Note also that in these coordinates, $\gamma$ is represented by the integral curve of $\partial/\partial_v$ which passes through the origin.
\end{theorem}

\begin{proof}
    First, define a vector field $Z$ (suggestively labelled in analogy to Theorem \ref{thm:adapted}) such that along $\gamma$ we have $Z = \dot{\gamma '}$. Now we construct the coordinate $u$. The partial differential equation 
    \[
        g(\text{grad}(u),\text{grad}(u)) = 0
    \]

    with boundary condition $\text{grad}(u) = Z$ on $\gamma '$ is a Hamilton-Jacobi equation for $u$ which always admits local solutions (see \cite[585-588]{Lee}).\\

    As is suggested by the similarity of the result, we take inspiration from the proof of Theorem \ref{thm:adapted}, noting that we no longer assume that $Z$ be covariantly constant. The necessary adjustment to the proof is as follows: That $Z$ is nonzero in a neighbourhood of $\gamma '$ holds again by the fact that it is null, but also by the fact that $\gamma '$ is geodesic, that is $\nabla _{\dot{\gamma'}} \dot{\gamma'} = 0$. Thus the remainder of step 1 remains valid, and we may construct a coordinate system $\{\tilde{x}^0,v,\tilde{x}^1,\ldots,\tilde{x}^{n-2}\}$ with $\text{grad}(u)=Z=\tilde{\partial}_v$ via the straightening theorem. Step 2 is not necessary in this context, as $u$ has already been introduced by the above argument. Step 3 follows as before, yielding a coordinate system $\{u,v,\mathbf{x}\} \vcentcolon = \{u,v,x^1,\dots,x^{n-2}\}$ on an open set $U \subset M $ containing $\gamma '$. The form of the metric and the positive-definiteness of $(g_{ab})$ then follow from parts (ii) and (iii) of Proposition~\ref{prop:metricProperties}.
\end{proof}

Note that an alternative and succinct version of this proof was provided by \cite[Sec. 4.3]{blau}, but the reader should note that their ``$U,V$" is our ``$v,u$".\\

We now describe the limiting procedure by which one can ``zoom in" on a null geodesic (called the \textit{Penrose limit}) while simultaneously scaling up the metric, in a manner analogous to obtaining the tangent space of a Riemannian manifold. The primary difference however is that in the Riemannian case, the space obtained via this procedure is a \textit{flat} space, whereas in the Penrose limit we will obtain an intrinsically \textit{curved} space, which will turn out to be the plane wave Eq. \ref{eq:plane} written in the Rosen coordinates of Sec. \ref{sec:rosen}.\\

\subsubsection{Limiting Procedure: Penrose's Construction}
This section follows Penrose's original construction \cite{penroseLimit} but is presented in a more modern language, in a self-contained manner using the proofs of Section \ref{sec:coordDesc}, and explicitly generalised to arbitrary dimension. The procedure by which we will define the Penrose limit of a spacetime will be (schematically) as follows:
\begin{enumerate}
    \item Take a spacetime $(M,g)$ and write the metric in null coordinates in a neighbourhood of a null geodesic $\gamma$.
    \item Define a new coordinate system whose coordinate functions are those of the null coordinates divided by powers of a parameter $\Omega$ (which we will let go to 0 later, causing those coordinates to ``blow up") and write $g$ in these coordinates.
    \item Define another metric $h$ on $M$ conformal to $g$ with constant factor $h = \Omega^{-2} g$
    \item Show that in the limit $\Omega \to 0$, $h$ (that is, $\Omega^{-2} g$) is simply the metric of a plane wave. This is the ``Penrose limit" of $(M,g)$ in a neighbourhood of $\gamma$, and importantly, the construction was independent of the properties of the spacetime metric $g$. That is, \textit{all} spacetimes look like a plane wave when we simultaneously scale up the coordinates and scale up the metric near a null geodesic $\gamma$, which amounts to ``zooming in" on $\gamma$, or equivalently, blowing up a neighbourhood of $\gamma$ to cover the whole spacetime.
\end{enumerate}

To understand the complementary scaling of the coordinates and the metric, Penrose interprets this procedure as first scaling up the coordinates to ``blow up" the points of interest (just as one does when looking at the tangent space of any point), then, to account for the fact that a general curvature tensor will appear to blow up as the coordinates do, we must simultaneously scale \textit{up} the metric to scale \textit{down} the curvature tensor and obtain finite results. Physically, Penrose interprets this procedure as boosting an observer closer and closer to the speed of light, and a complementary re-calibration of their clocks in such a manner so as to keep the affine parameter $u$ along the null geodesic $\gamma$ invariant under the procedure. For details see the original work \cite{penroseLimit} and \cite[Sec. 4.4]{blau} for a more modern description.\\

We now begin the explicit construction. Consider an $n$-dimensional Lorentzian manifold $(M,g)$ and an open set $U \subset M$ (containing a conjugate point-free segment of a null geodesic $\gamma$) on which the null coordinates Eq. \ref{eq:nullCoordsMatrix} are defined, and label this null coordinate chart $\psi$. Then consider the map $\phi_\Omega \vcentcolon=\varphi_{\Omega} \circ \psi \vcentcolon U \to \mathbb{R}^n$ where
\begin{align*}
\varphi_\Omega ~ \vcentcolon ~ &\mathbb{R}^n \to \mathbb{R}^n\\
\vcentcolon ~ &\left(u,v, x^{1}, \ldots, x^{n-2}\right) \mapsto \underbrace{\left(\frac{u}{\Omega^{2}}, v, \frac{x^{1}}{\Omega}, \ldots, \frac{x^{n-2}}{\Omega} \right)}_{=\left(\tilde{u}, \tilde{v}, \tilde{x}^{1}, \ldots, \tilde{x}^{n-2}\right) }
\end{align*}

for $\Omega > 0$ a constant. The map $\phi_\Omega$ is then a diffeomorphism for $\Omega \neq 0$. Define a metric\footnote{The metric $h$ will turn out to be conformal to $(\phi_\Omega^{-1})^*g$, but we could also start from that fact and define $h \vcentcolon= \Omega^{-2} (\phi_\Omega^{-1})^*g$ and the calculate its explicit form, which will be Eq. \ref{eq:hmatrix}.} $h$ on $\phi_\Omega(U) \subset  \mathbb{R}^n$ whose representation in the tilde coordinates is 
\begin{equation}\label{eq:hmatrix}
    h=
    \begin{pNiceMatrix}
        \Omega^{2} \tilde{H}       & 1         &   \Omega \tilde{A}_1                     &  \dots   & \Omega \tilde{A}_{n-2}        \\
        1       & 0         &   0                       &  \dots   & 0        \\
        \Omega \tilde{A}_1     & 0         & \Block{3-3}{(\tilde{g}_{ab})}  
                                                        &           &         \\
        \vdots  & \vdots    &                           &           &          \\
        \Omega \tilde{A}_{n-2}     & 0         &                           &           & 
    \end{pNiceMatrix},
\end{equation}
where $\tilde{H}$, the $\tilde{A}_a$ and the $\tilde{g}_{ab}$ are implicitly functions of \textbf{all} the tilde coordinates defined (strategically) in the following manner
\begin{align}
\begin{split}\label{eq:tildeFunctions}
    \tilde{H} &\vcentcolon = H(\Omega^2 \tilde{u}, \tilde{v}, \Omega\tilde{x}^1, \dots,\Omega\tilde{x}^{n-2}) = H(u,v,x^1,\dots,x^{n-2}),\\
    \tilde{A}_a &\vcentcolon= A_a(\Omega^2 \tilde{u}, \tilde{v}, \Omega\tilde{x}^1, \dots,\Omega\tilde{x}^{n-2}) = A_a(u,v,x^1,\dots,x^{n-2}),\\
    \tilde{g}_{ab} &\vcentcolon= g_{ab}(\Omega^2 \tilde{u}, \tilde{v}, \Omega\tilde{x}^1, \dots,\Omega\tilde{x}^{n-2}) = g_{ab}(u,v,x^1,\dots,x^{n-2}).
\end{split}
\end{align}
The metric $h$ is conformal to $(\phi_\Omega^{-1})^*g$, which can be seen as follows: First, by definition of the tilde coordinate system and Eq. \ref{eq:tildeFunctions}, we relate the components of $g$ and $h$ as:
\[
\begin{array}{c @{{}={}} c @{{}={}} c @{{}={}} c}
    g_{uv} ~ du \otimes dv & du \otimes dv & \Omega^2 d\tilde{u} \otimes d\tilde{v} & \Omega^2 h_{\tilde{u}\tilde{v}}~d\tilde{u} \otimes d\tilde{v},\\
    g_{uu} ~ du \otimes du & H du \otimes du & \Omega^{4} H d\tilde{u} \otimes d\tilde{u} & \Omega^{2} h_{\tilde{u}\tilde{u}} ~ d\tilde{u} \otimes d\tilde{u},
\end{array}
\]
and one obtains a similar relationship for the remaining components:
\begin{equation}
    g_{\rho\sigma} ~ dx^\rho \otimes dx^\sigma = \Omega^2 h_{\rho\sigma}~d\tilde{x}^\rho \otimes d\tilde{x}^\sigma.
\end{equation}

Second, since $\phi_\Omega$ is a change of coordinates, it holds that 
\begin{equation}
    g_{\rho\sigma}~dx^\rho \otimes dx^\sigma = ((\phi_\Omega^{-1})^*g)_{\rho\sigma}~d\tilde{x}^\rho \otimes d\tilde{x}^\sigma
\end{equation}

and thus 
\begin{equation}
     ((\phi_\Omega^{-1})^*g)_{\rho\sigma}~d\tilde{x}^\rho \otimes d\tilde{x}^\sigma = \Omega^2 h_{\rho\sigma}~d\tilde{x}^\rho \otimes d\tilde{x}^\sigma,
\end{equation}
that is, $h$ and $(\phi_\Omega^{-1})^*g$ are \textit{homothetic} (conformal with constant conformal factor) as
\begin{equation}
    h = \frac{1}{\Omega^2}(\phi_\Omega^{-1})^*g.
\end{equation}
We now actually take the \textit{Penrose limit} of $(M,g,\gamma)$, which is a neighbourhood of $\gamma$ in the spacetime formed by $M$ equipped with the metric 
\begin{equation}
    \lim_{\Omega \to 0} \frac{1}{\Omega^2}(\phi_\Omega^{-1})^*g = \lim_{\Omega \to 0}h.
\end{equation}

In this limit in the tilde coordinates, $h$ reduces to
\begin{equation}
    \lim_{\Omega \to 0}h=
    \begin{pNiceMatrix}
        0      & 1         &   0                     &  \dots   & 0        \\
        1       & 0         &   0                       &  \dots   & 0        \\
        0     & 0         & \Block{3-3}{(\tilde{g}_{ab})}  
                                                        &           &         \\
        \vdots  & \vdots    &                           &           &          \\
        0     & 0         &                           &           & 
    \end{pNiceMatrix},
\end{equation}
where $(\tilde{g}_{ab})$ is now a function of $\tilde{v} = v$ \textbf{only} as $\tilde{g}_{ab} = g_{ab}(0,\tilde{v}, 0,\dots,0)$. This is precisely the \textit{Rosen} coordinate representation of the plane wave metric Eq. \ref{eq:rosen} (under an appropriate relabelling/reordering of the coordinates).\\

What we have demonstrated is that in an appropriate limit around a null geodesic $\gamma$, any spacetime approaches a plane wave in a manner analogous to how a Riemannian manifold locally approaches Euclidean space in an appropriate limit. A collection of the Penrose limits of common spacetimes and a comprehensive overview of the properties of Penrose limits has already been established by \cite{blau}, such as the \textit{hereditary} properties (those properties of the limit which are inherited from the original spacetime). A covariant description of the limiting procedure is also provided, making significantly clearer the connection between the original metric $g$ and the properties of the resulting plane wave limit, which are encoded in the wave profile $H$ when written in the ``Brinkmann coordinates" as in Eq. \ref{eq:plane}.\\

We close our discussion of Penrose's limit by illustrating a family of examples.  These examples are taken from \cite[Eqn.~(3.1)]{blau2004penrose}, wherein full derivations can be found; here we write down only the resulting plane wave limit itself, restricting out attention to dimension 4.  Indeed, for both the Scharzschild metric and the Friedmann--Lemaître--Robertson--Walker (FRW) cosmological models, their Penrose plane wave limits take the following form in Brinkmann coordinates
\[
ds^2 = 2dudv + \sum_{a,b=1}^2\frac{A_{ab}x^ax^b}{u^2}du^2 + dx^2+dy^2,
\]
where each $A_{ab}$ is a constant depending on the original metric, and where $x^1=x$ and $x^2=y$.

%% file: Properties/Causal.tex
\subsection{Causality in Parallel Waves}\label{sec:causal}
We now review some basic results in the causal properties of parallel waves, starting with the well-known ``remarkable property of plane waves" proven by Penrose \cite{remarkable} which spurred on much of this research. 

\input{Properties/remarkable}

\subsubsection{Generic Position on the Causal Ladder}
After Penrose showed that the plane waves are not globally hyperbolic, interest was spurred in discovering the exact position of both the plane waves and pp-waves on the causal ladder. This question has been categorically answered for the plane waves by \cite{causalLadder4Waves}, and then for the $(N,h)$p-waves by \cite{causalLadderPFWs}. Note that the causality properties of the more general class of parallel waves does not appear to have been studied. Let us first recall the causal ladder for Lorentzian manifolds: 
\begin{center}
    Globally hyperbolic ($\exists$ a Cauchy surface)\\
    $\Downarrow$\\
    Causally simple (pasts and futures are closed + causality)\\
    $\Downarrow$\\
    Causally continuous (``continuity" of pasts and futures + distinguishing)\\
    $\Downarrow$\\
    Stably causal ($\exists$ a global time function)\\
    $\Downarrow$\\
    Strongly causal ($\nexists$ closed or ``almost closed" causal curves)\\
    $\Downarrow$\\
    Distinguishing ($\nexists$ points with same pasts and futures)\\
    $\Downarrow$\\
    Causal ($\nexists$ closed causal curves)\\
    $\Downarrow$\\
    Chronological ($\nexists$ closed timelike curves)\\
    $\Downarrow$\\
    Non–totally vicious ($\exists$ points $p\in M$ with $p\not\ll p$)
\end{center}
from \cite[Sec. 3]{causalLadderSanchez} and \cite{floresTalk}. Note that ``stably causal" was first understood as the causality being a stable property under perturbations, but Hawking showed \cite{Hawking1969cosmic} that this is equivalent to the existence of a global time function. Also note that $x \ll y$ means that $x$ \textit{chronologically precedes} $y$, that is there exists a future-directed chronological (timelike) curve from $x$ to $y$.\\

To make explicit our conventions, and to align with the conventions of \cite{causalLadderSanchez} we choose the signature of our spacetimes $(M,g)$ to be $(-,+,\dots,+)$, i.e., a non-zero vector field $X \in TM$ is 
\begin{itemize}\itemsep0em
    \item timelike $\iff$ $g(X,X) < 0$,
    \item lightlike $\iff$ $g(X,X) = 0$,
    \item spacelike $\iff$ $g(X,X) > 0$,
\end{itemize}
and we take the zero vector to be spacelike. We also use ``causal" to mean lightlike or timelike when referring to a vector field. Also to remain consistent with \cite{causalLadderPFWs}, when dealing with parallel waves we will fix our time-orientation such that $\partial_v$ is past-directed. We now examine the causal classification of the parallel waves, starting with the relatively simple result:
\begin{proposition}
    All $(N,h)$p-waves are chronological.
\end{proposition}
\begin{proof}
    For a parallel wave defined by a covariantly constant, null vector field $Z$, in the adapted coordinates of Theorem \ref{thm:adapted}, we have $Z = \nabla u = \partial_v$. For any future-directed causal curve $\gamma(s) = (u(s),v(s),\mathbf{x}(s))$ it holds that
    \[
        \dot{u}(s) = g(\dot{\gamma}(s),\partial_v) \geq 0
    \]
    where the inequality is sharp for $\gamma(s)$ timelike. Such an inequality prevents the existence of closed timelike curves, and thus the spacetime is \textit{chronological}.
\end{proof}
Being one of the ``lower rungs" of the causal ladder, being chronological is not a relatively strong restriction. We can however show that a generic $(N,h)$p-wave lies one step higher on the ladder:

\begin{theorem}\label{thm:causal}
    All $(N,h)$p-waves are causal. 
\end{theorem}
We will prove this theorem below using Proposition \ref{prop:quasi}. The proof of this result follows from \cite[Scholium 4.11]{causalLadder4Waves}, which we will reproduce here. To do so, we first introduce the concept of a \textit{quasi-time function}.
\begin{Definition}\label{def:quasi} \textup{Quasi-time function.}\\
    On a Lorentzian manifold $(M,g)$ a smooth function $f: M \mapsto \reals$ is called a \textit{quasi-time function} for $(M,g)$ if
    \begin{enumerate}[(i)]
        \item $\nabla f$ is everywhere nonzero, causal and past-directed, and if
        \item every null geodesic segment $\gamma$ such that $f \circ \gamma$ is constant,  is injective.
    \end{enumerate}
\end{Definition}
Now we may reproduce the afformentioned \cite[Scholium 4.11]{causalLadder4Waves} for completeness, which is stated as:
\begin{proposition}\label{prop:quasi}
    Any spacetime admitting a quasi-time function is causal.
\end{proposition}
\begin{proof}
    Assume $f$ is a quasi-time function as in Definition \ref{def:quasi}, then due to \textit{(i)} we have that $f$ is strictly increasing along all future-directed timelike curves in $M$, and hence $(M,g)$ is chronological. We now prove causality by contradiction.\\
    
    Assume $(M,g)$ is not causal, then $M$ would
    contain \cite[Scholium 4.10]{causalLadder4Waves} a non-trivial, smooth, future-directed null geodesic segment $\tilde{\gamma}:[0,1] \rightarrow M$ with $\tilde{\gamma}(0)=\tilde{\gamma}(1)$ and $\tilde{\gamma}^{\prime}(0)=\tilde{\gamma}^{\prime}(1)$.\\
    
    Furthermore $\tilde{\gamma}$ may be extended to an
    inextendible geodesic $\gamma: \reals \rightarrow M$ by letting $\gamma(s)=\tilde{\gamma}(s \bmod 1) .$ Again because of $(i)$ and by continuity of all the relevant properties, $f$ is non-decreasing along $\gamma$; hence $f \circ \gamma(s)=\lambda_{0}$ for all $s \in \reals$, constant $\lambda_0 \in \reals$, which would contradict $(ii)$, since $\gamma(0)=\gamma(1)$. Thus, $(M,g)$ must be causal.
\end{proof}

We now return to the proof of Theorem \ref{thm:causal}, armed with the knowledge of the above proposition. 
\begin{proof}\textbf{Proof of Theorem \ref{thm:causal}}\\
All that we require is that any $(N,h)$p-wave admits a quasi-time function. This is proven in \cite[Lemma 4.1]{causalLadder4Waves} and again is reproduced here. The claim is as follows:\\

\textbf{Claim:}
    When an $(N,h)$p-wave is written in the adapted coordinates of Theorem \ref{thm:adapted}, the coordinate function $u$ is a quasi-time function as in Definition \ref{def:quasi}.\\

To prove this, note that by definition we have a covariantly constant, null vector field $Z$ such that $Z = \nabla u = \partial_v$. Thus $\nabla u$ is causal by definition. Since $Z = \nabla u$ is nontrivial and covariantly constant, we have that $\nabla u$ is everywhere nonzero. Furthermore $\nabla u$ is past-directed since $\nabla u = \partial_v$ and the time-orientation on $(M,g)$ can be determined by the condition that $\partial_v$ be past-directed. Therefore point $(i)$ in the definition of a quasi-time function is satisfied.\\

Next, note that since the restriction of $g$ (Eq. \ref{eq:NPp}) to the null hypersurface $\Pi_{u_0} \vcentcolon = u^{-1}(u_0)$ for some $u_0 \in \reals$ is independent of the characteristic function $H$ and the wavefront is spacelike, the null geodesic segments will be of the form
\[
    \gamma: v\in \reals \mapsto (u_0,v,\mathbf{x}_0) \in \Pi_{u_0}.
\]
Such a map is injective, and thus point $(ii)$ in the definition of a quasi-time function also holds.
\end{proof}

\subsubsection{Conditions for Stronger Causal Character}
We now shift our focus to finding the conditions under which an $(N,h)$p-wave exhibits stronger causality properties. This was the subject of \cite{causalLadderSanchez}, in which is was shown that the criterion for determining causal character is the spatial asymptotic behaviour of the characteristic function $H$ (when the parallel wave is written in adapted coordinates), and in some cases the completeness of the Riemannian manifold corresponding to the wavefront. A summary of the results of this work \cite[Sec. 7]{causalLadderPFWs} is given in Table \ref{tab:PFWCausal}, where one uses $-H$ to classify asymptotic behaviour as opposed to $H$ to be consistent with work which will be presented in Sec. \ref{sec:EKConj}. A precise definition of the asymptotic behaviour of $H$ follows from:
\begin{Definition} \textup{Subquadratic Growth.}\\
    We say that $-H(u,\mathbf{x})$ behaves \textit{subquadratically at spatial infinity} if there exists some $\mathbf{x}_0\in N$ (where $N$ is the wavefront) and continuous functions $R_1(u), R_2(u)(\geq 0), p(u) < 2$ such that:
    \[
    -H(\mathbf{x}, u) \leq R_{1}(u) d^{p(u)}(\mathbf{x}, \mathbf{x}_0)+R_{2}(u) \quad \forall ~(u,\mathbf{x}) \in \reals \times N,
    \]
where $d$ is the distance canonically associated to the Riemannian metric on $N$. When $p(u) \equiv 2$, then we say $-H(u,\mathbf{x})$ behaves (at most) quadratically at spatial infinity\footnote{For the sake of completeness, one would similarly define \textit{superquadratic growth} via $-H(\mathbf{x}, u) > R_{1}(u) d^{p(u)}(\mathbf{x}, \mathbf{x}_0)+R_{2}(u) \quad \forall ~(u,\mathbf{x}) \in \reals \times N$.}.
\end{Definition}
\begin{table}[]
\begin{tabular}{llll}
\hline
Condition on $H$                                                                                                 &            & Causal Character    & $\exists$ Examples                                                                                                     \\ \hline
$-H$ Superquadratic                                                                                              & $\implies$ & Causal              & \begin{tabular}[c]{@{}l@{}}Non-distinguishing\\ \& globally hyperbolic\end{tabular}      \\ \hline
$-H$ Quadratic                                                                                                   & $\implies$ & Strongly Causal     & \begin{tabular}[c]{@{}l@{}}Globally hyperbolic\\ \& non-globally hyperbolic\end{tabular} \\ \hline
\begin{tabular}[c]{@{}l@{}}$-H$ Subquadratic \\ \& wavefront complete\end{tabular} & $\implies$ & Globally Hyperbolic &                                                                                                                    \\ \hline
\end{tabular}
\caption{\label{tab:PFWCausal}Causal properties of an $(N,h)$p-wave under certain conditions on the characteristic function $H$. The rightmost column lists the (non-generic) causal character of certain examples with the corresponding asymptotic behaviour of $-H$. Results as in \cite{causalLadderPFWs}.}
\end{table}

In light of Table \ref{tab:PFWCausal} we can identify $H$  being quadratic as critical for the causal behaviour, in the sense that small perturbations either in the superquadratic or in the subquadratic direction may
introduce significative qualitative differences in the causal character.\\

%% file: Properties/remarkable.tex
\subsubsection{A Remarkable Property of Plane Waves}\label{sec:remarkable}
Roughly, Penrose showed that a (not necessarily purely gravitational) plane wave exhibits a ``focusing property" on the null cones (see Fig. \ref{fig:remarkable}), and as a consequence, there exists no Cauchy hypersurface sufficient for the specification of Cauchy data \cite{remarkable}. This is because the past null cone of any event is focused to a single point (anastygmatism) or line (astygmatism), and since a Cauchy hypersurface has the property that it
intersects any causal curve exactly once, it is concluded that this focusing property forces many causal curves to intersect any potential Cauchy hypersurface at least twice. In the following, we maintain consistency with the notation of the original work wherever possible.\\

To begin, let us first define the relevant objects. As in Sec. \ref{sec:plane} (with a small relabelling), a plane wave is defined as a 4-dimensional standard pp-wave in adapted coordinates $\{u,v,x^1,x^2\}$ for which the characteristic function $H(u,x^1,x^2)$ is \textit{quadratic} in $(x^1,x^2)$, that is the spacetime $(M =\reals^4,g)$ where
\begin{align*}
    &g = 2\mathrm{~d} u \mathrm{~d} v + H(u,x^1,x^2)~\mathrm{d}u^2 + (\mathrm{d} x^1)^{2}+(\mathrm{d} x^2)^{2}\\
    &H(u,x^1,x^2) = \sum_{i, j=1}^{2} h_{i j}(u) x^{i} x^{j}
\end{align*}
for some symmetric matrix formed by the $h_{ij}$. We also define the \textit{null cone}:

\begin{Definition} \textup{Null Cone}\\
    The null cone (denoted $\kappa_3$) at a point $Q \in M$ is defined as the set of points lying on all null geodesics through $Q$.
\end{Definition}

In this section, Penrose utilises the so-called ``sandwich waves", defined by the characteristic that the amplitudes $h_{ij}(u) = 0$ unless $u \in (a,b)\subset \reals$. One can visualise such a plane wave as in Fig. \ref{fig:remarkable}, in which it becomes clear that a sandwich wave is a plane wave for which the infinite extent in the $u$ direction is removed. \\

\begin{figure}
    \centering
    \makebox[0.5\textwidth][c]{%
    \begin{subfigure}{.5\textwidth}
      \centering
      \includegraphics[width=0.9\linewidth]{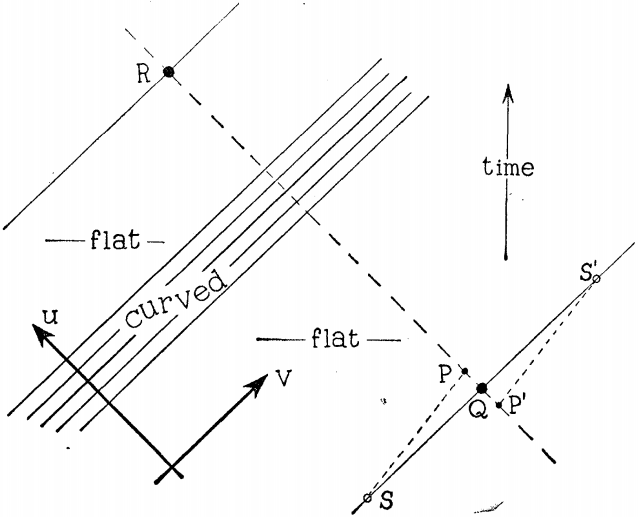}
    \end{subfigure}%
    \begin{subfigure}{.5\textwidth}
      \centering
      \includegraphics[width=1\linewidth]{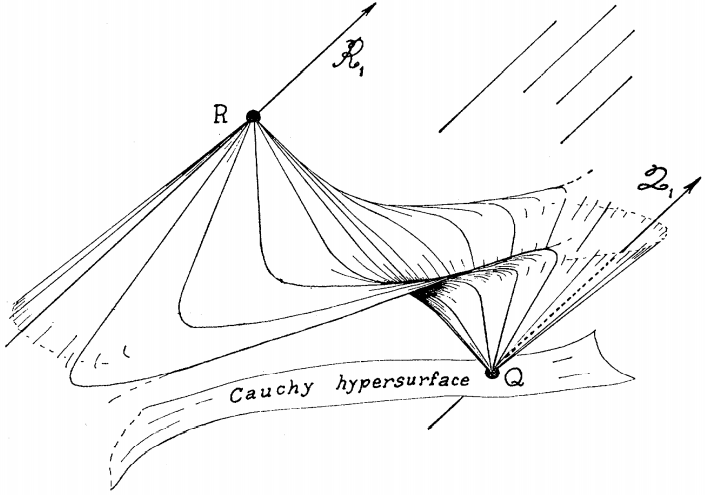}
    \end{subfigure}
    }%
    \caption{(Left) The wave profile of a sandwich plane wave, in which the $u$ coordinate range of the ``curved" region is $(a,b)$. (Right) The focusing effect of such a wave on the past null cone of a point R in the electromagnetic case. Figures from \cite[Fig. 1 \& 2]{remarkable}.}
    \label{fig:remarkable}
\end{figure}

We now outline the primary result of \cite{remarkable}, where some details are omitted and only the main steps of the proof are reproduced.
\begin{theorem}
    The past null cone of any point $Q$ in a plane wave $(M,g)$ with compactly supported profile (a ``sandwich wave") is focused to a single point for an electromagnetic sandwich wave, or to a line for a gravitational sandwich wave.
\end{theorem}
\begin{proof}
    To begin, choose a point $Q$ in the flat region of $M$, such that the components of $Q$ are
    \[
        u = u_0 < a, \quad v = v_0, \quad x^i = 0,
    \]
    where $a$ is the lower bound of the interval on which $u$ is nonzero for the sandwich wave. Close to $Q$, the equation of the null cone $\kappa_3$ is $(u - u_0)(v - v_0) - x^ix^i = 0$ which can be written
    \begin{equation}\label{eq:fInit}
        v = f_{ij}(u) x^ix^j + v_0,
    \end{equation}
    where $f_{ij}(u) = (u - u_0)^{-1}\delta_{ij}$ near $Q$. We now wish to obtain a description of $\kappa_3$ valid away from $Q$, that is to find an appropriate $f_{ij}(u)$. If the surface is to remain \textit{null} even in the curved regions of $M$, then one can show that $f_{ij}$ should be both symmetric and satisfy\footnote{The original paper lists the condition as $\frac{d}{du}f_{ij} + f_{ik}f_{kl} + h_{ij} = 0$, the $l$ index likely being erroneous.}
    \begin{equation}\label{eq:diffF}
        \frac{d}{du}f_{ij} + f_{ik}f_{kj} + h_{ij} = 0.
    \end{equation}
    With ``initial condition" Eq. \ref{eq:fInit} one obtains an $f_{ij}$ which describes the null cone $\kappa_3$ even in the curved region of $(M,g)$. This extension is only valid while $f_{ij}$ is finite, and so we now examine if and when $f_{ij} \rightarrow \infty$. To do so, consider the trace of the above differential equation, noting that $h_{ij}$ is trace-free for a vacuum solution and in general $h_{ii} > 0$.
    \[
        \frac{d}{du}f_{ii} + \frac{1}{2}f_{ii}f_{jj} = -\frac{1}{2}\left(f_{i k} f_{i k} \delta_{j l} \delta_{j l}-f_{i k} \delta_{i k} f_{j l} \delta_{j l}\right)-h_{i i} \leq 0
    \]
    via Schwarz' inequality. Defining $\rho(u) \vcentcolon=\frac{1}{2} \int_{u_0}^u f_{ii}(\bar{u})d\bar{u}$, one finds the integro-differential inequality on the trace of $f$
    \begin{equation}\label{eq:integro}
        \frac{d^2}{du^2} \rho(u) \leq 0,
    \end{equation}
    where the inequality is sharp for at least some values of $u$. Since our choice of $u_0$ in $Q$ was arbitrary, consider the limit $u_0\longrightarrow -\infty$. Then from the definition of $f_{ij}$ near $Q$, we see that $f_{ij} = 0 ~\forall~ u < a$. Then via Eq. \ref{eq:fInit}, we see that $\kappa_3$ is described by the equation $v = v_0$, that is the null cone is a null hyperplane in the flat region. When $f_{ij} = 0$ then in particular $\rho' = 0$ in the flat region (prime meaning $u$-derivative), and therefore by Eq. \ref{eq:integro} we have that a $\rho$ which is positive in the flat region near $Q$ will become 0 for finite $u$. If $\rho =0 $ then some component of $f_{ij}$ must become singular\footnote{This is related to the fact that there are often coordinate singularities when writing a pp-wave metric in Rosen coordinates, which originally lead to the belief that there did not exist non-singular plane wave solutions of the full Einstein equations. See \cite[footnotes 11,12]{remarkable} for details.}. Denote the $u$ at which $f_{ij}$ exhibits singularity by $u_1 > a$ (since for $u_1 \leq a$ we have $f_{ij} \equiv 0$).\\
    
    If this singularity occurs outside the curved region, i.e. $u_1 > b$ then the null cone $\kappa_3$ encounters a singularity on the ``past" side of the sandwich wave. In fact, one needs to consider large and negative $u_0$ as opposed to the $-\infty$ limit, but this does not affect the relevant equations here.\\
    
    Now consider the flat region containing this singularity. In this region Eq. \ref{eq:diffF} may be written as $p_{ij}' = \delta_{ij}$ where $p_{ij}$ is the inverse\footnote{In the original reference this is written as $p_{ij}f_{jk} = \delta_{ij}$, again likely to be erroneous.} matrix to $f_{ij}$, i.e. $p_{ij}f_{jk} = \delta_{ik}$. The solution of this differential equation for $p_{ij}$ is
    \[
        p_{ij}(u) = u\delta_{ij} - q_{ij}
    \]
    for constant and symmetric $q_{ij}$ (since $f$ is symmetric). Therefore $f_{ij}$ has a singularity whenever $u$ is an eigenvalue for $q_{ij}$. Either these eigenvalues are distinct or they are degenerate, in which case $q_{ij} = u_1\delta_{ij}$. In this degenerate case, $p_{ij}$ has the form $(u - u_1)\delta_{ij}$, and $\kappa_3$ has two vertices, namely $P$ and the point $R \vcentcolon = (u_1,v_0,\mathbf{0})$. This is because the equation of $\kappa_3$ reduces to a single point at both $P$ and $R$, as in fig.\ref{fig:remarkable}. In fact, that $\kappa_3$ is focused to a single \textit{point} (anastygmatic) is specific to the purely electromagnetic case in which $h_{ij}$ is purely diagonal. For the gravitational case, one finds that $\kappa_3$ is focused onto a line. Since the arguments used are very similar, we omit this proof here. See \cite{remarkable} for details.
\end{proof}
To explain why this result shows that plane waves are not globally hyperbolic, consider a candidate for a Cauchy hypersurface. Such a hypersurface would have to
intersect the $v$-line through $R$. But then some of the other past-oriented lightlike
geodesics from $R$ to $Q$ have to be intersected twice. Looking to Fig. \ref{fig:remarkable}, a connected spacelike hypersurface such as the proposed Cauchy hypersurface containing $Q$ must initially lie entirely in the past of (drawn as ``below" on the diagram) the future null cone of $Q$. A Cauchy hypersurface can never meet the null line $\mathscr{R}_1$, as if it were to do so then it would intersect the null geodesics through $Q$ twice (since they are all focused onto $\mathscr{R}_1$). As a result, the proposed Cauchy hypersurface must ``bend downwards" to avoid $\mathscr{R}_1$, and can never extend through it while remaining everywhere spacelike, and as in \cite{remarkable}: ``Cauchy data on such a hypersurface could thus give no information for specifying amplitudes for a parallel wave\footnote{Curiously, this appears to be one of the first uses of the name ``parallel wave". Note however that this phrasing is not consistent with the definitions of this article, and the object in question is more accurately referred to as a ``plane wave".} which might lie beyond $\mathscr{R}_1$".

%% file: EKConj.tex
\section{The Ehlers--Kundt Conjecture}\label{sec:EKConj}
The Ehlers--Kundt conjecture is a statement about the role of gravitational plane waves (Eq. \ref{eq:plane}) in the mathematical description of gravitational waves. Roughly, it claims that the plane waves act as a mathematical idealisation of gravitational waves, and was originally stated as follows:
\begin{quotation}
    \textit{``Prove the plane waves to be the only complete pp-waves."\footnote{The Ehlers--Kundt conjecture originally contained the addendum ``no matter which topology one chooses", but as discussed in \cite{EKConj} the extension of the conjecture to manifolds of general topology is nontrivial. This extension was provided by \cite{compact}, which reduces to the statement above under the appropriate conditions.}}
\end{quotation}
The conjecture can be stated in a more modern language as follows, where the terms ``plane wave" and ``classical pp-wave" are defined consistently with the nomenclature of \textit{this} article (see Table \ref{tab:definitions}):
\begin{quotation}
    \textit{``Prove the plane waves to be the only geodesically complete, Ricci-flat classical pp-waves."}
\end{quotation}

    The conjecture stems from the idea that gravitational radiation should not arise in a spacetime in which there is no source to create it. If a spacetime is complete and Ricci-flat\footnote{For clarity, Ricci-flat = purely gravitational = vacuum = no matter present.} but the metric describes a propagating wave, then that wave would be produced independent of any source. Since complete spacetimes are inextendible, that is they are not part of some larger spacetime, we can be sure that we are not just ``missing" the part of the spacetime containing a source. If a vacuum spacetime contains a wave but is not complete, it is certainly possible that we are missing the source in our description.\\
    
    An analogy would be a room with light coming from behind a curtain. In this analogy light is the pp-wave, ``vacuum" means we cant see any lightbulbs (sources), and completeness equates to removing the curtain, so we can see everywhere in the room. If the curtain is present and we see light in the room, it is reasonable to say there must be a source behind the curtain. However it seems impossible that there is light in the room, we \textit{can} see everywhere, \textit{and} there is no lightbulb. To translate back to our terminology, it seems it should be impossible that our spacetime contains a wave,  is complete, and is also Ricci-flat.\\
    
    \begin{figure}
    \centering
    \makebox[0.6\textwidth][c]{%
    \begin{subfigure}{.6\textwidth}
      \centering
      \includegraphics[width=1\linewidth]{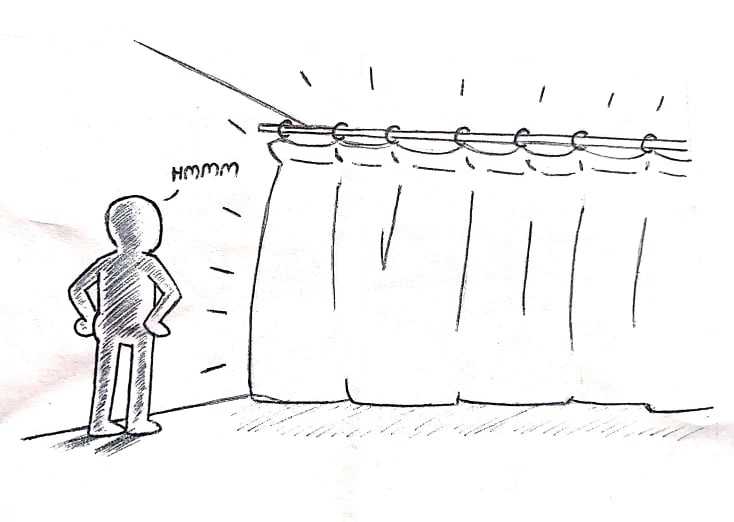}
    \end{subfigure}%
    \begin{subfigure}{.5\textwidth}
      \centering
      \includegraphics[width=1\linewidth]{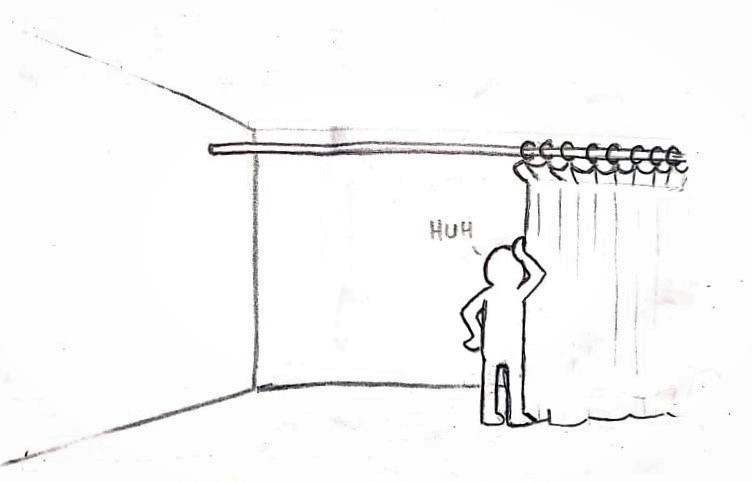}
    \end{subfigure}
    }%
    \caption{An analogy for the Ehlers--Kundt conjecture. Art courtesy of Christopher Martin.}
    \label{fig:EKChris}
\end{figure}

    Ehlers and Kundt \cite{exactsolEK} showed that the plane waves are always complete, even in the vacuum case. That is they correspond to the apparently unphysical case of a lit room with no curtain and no lightbulb. The Ehlers--Kundt conjecture assigns the plane waves the role of mathematical idealisations, and claims that any other pp-wave (\ref{eq:generalPp}) must be incomplete, so that the source which ``must have" created the waves is simply not part of our description. This is strongly related to the fact proven by \cite{remarkable}, wherein Penrose shows that the 
    plane waves are \textit{not globally hyperbolic}, as discussed in Sec.~\ref{sec:remarkable}.
    
    Spacetimes which are both\footnote{It is necessary to have both completeness and non-global hyperbolicity to claim the spacetime is unphysical. This is because by removing a point from a globally hyperbolic spacetime, one ``destroys" that global hyperbolicity. If that spacetime can be extended (here, by adding that point back) to a globally hyperbolic one, we should not consider it necessarily unphysical. However, a complete spacetime is inextendible, meaning there is no possibility to ``get back" the global hyperbolicity. For this reason, if a non-globally hyperbolic spacetime is complete, we can safely consider it unphysical.} complete and not globally hyperbolic are generally considered unphysical, since the development of the spacetime from arbitrary initial data in the initial value formulation of the Einstein equations is not unique in this case. %
    This construction is outlined in section \ref{sec:remarkable}. The EK-conjecture for gravitational pp-waves can be summarised as ``spacetime is complete" $\iff$ it is a plane wave. However since the $\impliedby$ direction was already proven by \cite{exactsolEK}, the conjecture in fact only refers to the $\implies$ direction.\\

    Although there is no known counterexample (i.e. a complete classical pp-wave other than the plane wave), the conjecture remains an open question. Significant progress has been made in addressing it however, and the remainder of this section will outline that progress. To begin, let us formulate the conjecture in more precise mathematical terms, and focus our attention on the classical pp-waves on $M = \reals^4$ so that our metric takes the form
\begin{equation}
    g = 2 d u d v - V(u, x,y) d u^{2}+ dx^2 + dy^2,
\end{equation}
where to be Ricci-flat/vaccum we must have that $V \vcentcolon = -H$ is harmonic in $(x,y)$. That is, $V_{xx} + V_{yy}=0$. 
The Ehlers--Kundt conjecture in this case states: if $(M,g)$ is geodesically complete, then $V(u,x,y)$ must be quadratic in $(x,y)$. We may replace the ``complete" in the original statement with ``geodesically complete" and study the geodesic equations of $(M,g)$. Upon calculating the geodesic equations, one finds
\begin{align}
        \ddot{u} &=0 \\
        \ddot{v} &=\frac{\dot{u}}{2}\left(\dot{u} V_{u}(u, x, y)+2 \dot{x} V_{x}(u,     x, y)+2 \dot{y} V_{y}(u, x, y)\right) \\
        \ddot{x} &=-\frac{\dot{u}^{2}}{2} V_{x}(u, x, y) \\
        \ddot{y} &=-\frac{\dot{u}^{2}}{2} V_{y}(u, x, y),
\end{align}
where a dot represents the derivative with respect to an affine parameter\footnote{Note that since the solution for $u(t)$ is $at+b$ for constants $a$ and $b$, then $u$ can be used as an affine parameter along the geodesic. This fact extends also to $n$ dimensions and does not depend on the properties of $H$.} $t$. Since the boundary conditions determine $u$ entirely, and the completeness of $v(t)$ evidently depends only on the completeness of $x(t)$ and $y(t)$, in studying the completeness the geodesic equations reduce to
\begin{equation}
\begin{split}
\ddot x(u) &= - V_x(u,x,y),\\
\ddot y(u) &= - V_y(u,x,y).
\end{split}
\end{equation}

These equations can be recast as a Hamiltonian system by defining $q(u) = (x(u),y(u))$, $p = \dot{q}$, and $\nabla$ the Euclidean gradient on $\reals^2$, such that we have 
\begin{equation}\label{eq:dynamicalSys}
    \dot{p} = -\nabla V(u,q).
\end{equation}

In this section we will use only $V$ as opposed to $H$, in order to maintain the interpretation as the potential of a dynamical system in classical mechanics. The Ehlers--Kundt conjecture can be restated in this language as: Prove that for $V(u,x,y)$ harmonic in $(x,y)$, if the Hamiltonian system $\dot{p} = -\nabla V(u,q)$ admits global solutions for all initial data, then the $u$-constant function $V(u,\cdot)$ is an at most quadratic polynomial in $(x,y)$. As mentioned above, this statement has not been proven in general. Before moving on to examine the special cases in which the conjecture have been proven, beginning with the so-called \textit{polynomial EK-conjecture}, we pause to mention a beautiful connection this conjecture has with complex dynamics, an observation due to G. Cox (private communication).

\subsection{Relation to Complex Dynamics}
In what follows, assume that $V$ is independent of $u$ (``autonomous"), and consider the complex-valued function $f\colon \mathbb{C} \rightarrow \mathbb{C}$ constructed from the partial derivatives $V_x,V_y$ of $V$:
\begin{equation}
\label{eq:Hforst}
z = x+iy \hspace{.2in},\hspace{.2in}f(z) = -V_x(x,y) + i V_y(x,y).
\end{equation}
The Cauchy-Riemann equations are
\[
-V_{xx} = V_{yy} \hspace{.2in},\hspace{.2in} -V_{xy} = -V_{yx},
\]
and observe that, while the second equation holds trivially, the first equation is satisfied precisely when $V(x,y)$ is harmonic (this is also the case for $f(z) = V_y + i V_x$).  It was shown in \cite[Corollary 7.4]{forstneric} that, given any entire function $f(z)$ (i.e., a function holomorphic on the entire complex plane $\mathbb{C}$), the complex-valued ODE
\[
\ddot{z} = f(z)
\]
admits global solutions for all initial data if and only if $f(z)$ is affine linear.  If we apply this result to Eq. \ref{eq:Hforst}, one finds
\begin{equation}
\label{eq:Hforst2}
\ddot{x} + i\ddot{y} = \ddot{z} = f(z) = -V_x + iV_y,
\end{equation}
then \cite[Corollary 7.4]{forstneric} yields that this system is complete if and only if $V_{xxx} = V_{yyy} = 0$; i.e., if and only if $V$ is quadratic in $x,y$.  This is not quite a proof of the EK-conjecture, however, since the pair of real ODEs to which Eq. \ref{eq:Hforst2} gives rise is not the usual Hamiltonian system Eq. \ref{eq:dynamicalSys}, but rather the following variation of it:
\[
\ddot{x} = -V_x \hspace{.2in},\hspace{.2in} \ddot{x} = V_y.
\]
Indeed, to obtain the usual Hamiltonian ODEs we should have chosen instead the function
\[
f(z) = -V_x - iV_y.
\]
(See also Eq. \ref{eq:autonomousV} in Remark 5.1 below.) Unfortunately, this function is holomorphic if and only if the harmonic function $V$ is linear; indeed, owing to Eq. \ref{eq:Hforst}, this choice of $f(z)$ is precisely \emph{anti}-holomorphic (i.e., its complex-conjugate is holomorphic).  We therefore come to the beautiful realization that the EK conjecture is the anti-holomorphic analogue of \cite[Corollary 7.4]{forstneric} and, as such, forms a bridge connecting general relativity to complex dynamics.  The main ingredient in the proof of \cite[Corollary 7.4]{forstneric} is a classification of the complete complex orbits of $\ddot{z} = f(z)$ which shows that they must be isomorphic to certain Riemann surfaces \cite[Proposition 3.2]{forstneric}; it is an intriguing question to see if the complete orbits of Eq. \ref{eq:dynamicalSys}, in the case when $V$ is harmonic, can be similarly classified.

\subsection{Polynomial EK-Conjecture}
In this section we will outline some of the work done by Flores and S\'anchez in \cite{EKConj}, who studied the EK-conjecture in the case that the potential $V$ is polynomially bounded. We refer to the case when $V$ does not depend on $u$ as the ``autonomous case", that is $V = V(x,y)$. The $u$-dependence of $V$ is not restricted by any of the previous discussion, and so it is natural to first consider the autonomous case. To make statements about the completeness of trajectories, the authors make use of \textit{confinement} properties of the relevant ODEs, and so we begin by developing some intuition for this:

\subsubsection{Motivation for Proof}
As a point of entry into thinking about the Ehlers--Kundt conjecture, consider for a moment the case when $V$ is an autonomous harmonic polynomial that is \emph{even} in $y$, namely, $V(x,-y) = V(x,y)$; e.g.,
\begin{equation}
\label{eqn:pos}
V(x,y) = -x^3+3xy^2\hspace{.1in}\text{and}\hspace{.1in}V(x,y) = -x^4+6x^2y^2-y^4
\end{equation}
are two such examples. The virtue of this class of harmonic polynomials is that, since the partial derivative $V_y$ is necessarily \emph{odd} in $y$, we must have $V_y(x,0) = 0$.  As a consequence, the ODE
\[
\ddot{y} = -V_y(x(t),y(t))
\]
admits the trivial solution $y(t) = 0$, for which choice the remaining ODE in $x$ takes the form
\begin{equation}
\label{eqn:HH2}
\ddot{x} = -V_x(x(t),0).
\end{equation}
Any solution $x(t)$ to Eq. \ref{eqn:HH2} then yields a solution $(x(t),0)$ of our original two-dimensional ODE\,---\,and the advantage to this approach is that Eq. \ref{eqn:HH2} permits a much easier blow-up analysis.  Indeed, consider any autonomous harmonic polynomial that is not even in $y$, but, like the examples in Eq. \ref{eqn:pos}, has \emph{negative} leading term in $x$:\footnote{In fact any harmonic polynomial that is even in $y$ can be put in such a form by a rotation of the $xy$-plane, where we note that rotations are isometries of the pp-wave metric, and that they also preserve the property of being harmonic.}
\[
V(x,0) = - (a_dx^d + a_{d-1}x^{d-1} + \cdots + a_1x + a_0)\comma a_d > 0\ ,\ d \geq 3.
\]
Then, since $a_d > 0$, we can, by a translation $x \mapsto x + a$ if necessary (which is an isometry of the standard pp-wave metric), assume that each $a_i \geq 0$ as well.  But now with ``every term negative", it follows easily that the solution $x(t)$ to Eq. \ref{eqn:HH2} satisfying $x(0) = 1$ and $\dot{x}(0) = \sqrt{2a_d}$ must be bounded \emph{above} (i.e. bounded \textit{below} in absolute value) by the corresponding solution to 
\[
\bar{V}(x) = -a_dx^d  \comma \ddot{x} = \bar{V}'(x(t)) = -da_dx(t)^{d-1}.
\]
since $V < \bar{V}$. This latter, bounding solution is
\[
\bar{x}(t) = \frac{b}{(c-t)^{\frac{2}{d-2}}} \comma b \defeq \Big[\underbrace{\frac{2}{a_d(d-2)^2}}_{>\,0}\Big]^{\!\frac{1}{d-2}} \comma c \defeq b^{\frac{d-2}{2}},
\]
which blows up in finite time.  Thus, since $|x(t)|$ is bounded below by a function that blows up in finite time, it follows that the solution $(x(t),0)$ also blows up in finite time.\\

What made this approach work?  It was the property of being even in $y$ that allowed us to find geodesics that \emph{stay in a confined region of the $xy$-plane\,---\,namely, the $x$-axis\,---\,which confinement simplified the resulting ODEs to the point where their behavior was dominated by the leading term of just one polynomial}.  This is an effective means of symplifying the analysis, but, of course, not every harmonic polynomial is even in $y$.  The questions remains, therefore, as to whether this technique of ``concentrating in a particular region of the plane" can work in general. Indeed it was demonstrated in \cite{EKConj} that this technique \emph{does} work in full generality, thereby resolving the polynomial case of the Ehlers--Kundt conjecture.\\

\begin{Remark}
In their work on the polynomial case of the EK-conjecture \cite{EKConj} the authors use a complex variable approach, wherein $z \vcentcolon = x + iy$ takes the place of the vector $q$ and similarly $\dot{z} = p$. There is a good reason that we should consider the polynomial case in the complex numbers $\mathbb{C}$ as opposed to the real numbers. As explained in \cite[pg. 5]{EKConj}, in the autonomous 
case $V:\reals^2 \rightarrow \reals$ we may identify $\mathbb{C}$ with $\reals^2$. The completeness of the trajectories of a potential $V$ is equivalent to the completeness of a corresponding vector field $X$ 
on the tangent bundle, and there exists a well-established
theory about completeness of holomorphic
vector fields $X$ on $\mathbb{C}^2$ \textit{in the case that they are polynomial}. The more general case where $V$ is not polynomially bounded does not admit an obvious advantage in the complex language. In this notation, the geodesic equations take the form
\begin{equation}\label{eq:autonomousV}
    \dot{p} = -\nabla V(q) \implies \ddot{z} = -V_x(x,y) - iV_y(x,y)
\end{equation}
For the purposes of this review, we will continue to explicitly write $x$ and $y$ in place of $z$.
\end{Remark}

We now ask ourselves if the above ODE Eq. \ref{eq:autonomousV} admits global solutions for $V$ harmonic in $(x,y)$, that is we wonder if the corresponding spacetime manifold in the original statement of the EK conjecture is geodesically complete. In fact, this is an open question in general. The following partial result by \cite{candelaCompleteness} became an important motivation for the so-called \textit{polynomial EK-conjecture}:

\begin{theorem}(Candela, Romero \& S\'anchez ’13)\\
For $V: \reals^2 \rightarrow \reals$ harmonic in $q \vcentcolon = (x,y) \in \reals^2$, if there is a constant $b\in \reals$ such that $V(q) \geq - b|q|^2$ for all $q \in \reals^2$, then the ODE $\ddot{q} = -\nabla V(q)$ admits global solutions for all initial data.
\end{theorem}
In other words, this is the statement that the Ehlers--Kundt conjecture holds in the case that $H = -V$ is \textit{subquadratic}. We reproduce now a short version of the proof which is originally due to G. Cox (private communication):
\begin{proof}
It is sufficient to assume $b > 0$. Since we have translated the original conjecture to the realm of Newtonian dynamics, we may apply simple energy conservation
\begin{equation}
    \frac{1}{2}|p|^2 + V(q) = E \Rightarrow |p|^2 \leq 2(E + b|q|^2).
\end{equation}
We then bound $|p|$ by $|q|$ in the cases of negative and non-negative energy: 
\begin{equation*}
\begin{aligned}
E<0 & \Rightarrow|p|^{2} \leq 2 b|q|^{2}, \\
E \geq 0 & \Rightarrow 2\left(E+b|q|^{2}\right)=\underbrace{2(\sqrt{E}+\sqrt{b}|q|)^{2}-4 \sqrt{E b}|q|}_{\leq 2(\sqrt{E}+\sqrt{b}|q|)^{2}}. \\
\end{aligned}
\end{equation*}
Such that in both cases we have the bound
\begin{equation}
    |p| \leq a+c|q| \quad, \quad a \geq 0, c>0.
\end{equation}
We can then bound $|q(t)|$ using $|q(0)|$ as follows:
\begin{equation*}
    \begin{aligned}
    \underbrace{\left|\int_{0}^{t} p(s) d s\right|}_{|q(t)|-|q(0)| \leq} & \leq \int_{0}^{t}|p(s)| d s \leq \underbrace{\int_{0}^{t}(a+c|q(s)|) d s}_{a t+c \int_{0}^{t}|q(s)| d s} \\
    & \Rightarrow|q(t)| \leq(|q(0)|+a t)+c \int_{0}^{t}|q(s)| d s \\
    & \Rightarrow|q(t)| \leq \underbrace{(|q(0)|+a t) e^{c t}}_{\text {bounded on compact int. }}
    \end{aligned}
\end{equation*}
where in the final step we have used the integral form of Gr\"onwall's inequality. The result then follows by Picard-Lindel\"of. 
\end{proof}

This result was proven in \cite{candelaCompleteness} even in the case that $V$ is \textit{non-autonomous} and where $|\cdot|$ is replaced by a general distance function $d_g(\cdot~,\cdot)$ associated to a Riemannian metric $g$. Therefore the previous result also holds true for a gravitational $(N,h)$-fronted wave \ref{eq:NPp}. That the EK-conjecture is true for a harmonic and \textit{subquadratic} $H = -V$ motivates one to ask if the same is true for harmonic and \textit{polynomially bounded} $H$. This question was answered by \cite{EKConj}, but before stating the theorem let us first make precise the idea of a \textit{polynomially bounded} $H$.
\begin{Remark}
Following the terminology of \cite{EKConj}, a function $H:\reals \times \reals^2 \rightarrow \reals$ is called ``polynomially $u$-bounded"  (meaning polynomially upper bounded along finite $u$-times) when for each $u_{0} \in \mathbb{R}$, there exists $\epsilon_{0}>0$ and a polynomial $P_{0}:\reals^{2} \rightarrow \reals$ such that $H(u,q) \leq P_{0}(q)$ for all $(u,q) \in \left(u_{0}-\epsilon_{0}, u_{0}+\epsilon_{0}\right) \times \mathbb{R}^{2}$.
\end{Remark}
Note that we say H is quadratically polynomially $u$-bounded when $P_0$ can be chosen of degree 2 for all $u_0 \in \reals$.

\subsubsection{Outline of Proof}
The Polynomial EK-conjecture is stated as follows:

\begin{theorem}\label{thm:polyEK}(Flores \& S\'anchez ’19)\\
Let $V:\reals \times \reals^2 \rightarrow \reals$ be a polynomially $u$-bounded $C^1$-potential
which is also $C^2$ and harmonic in the pair of variables q = (x, y). Then: all
the solutions to the dynamical system Eq. \ref{eq:dynamicalSys} are complete if and only if the function
$V (u,\cdot)$ is an at most quadratic polynomial for each $u \in \reals$.
\end{theorem}
We will present here only a rough outline of the arguments behind the proof, following loosely \cite[Sec. 2.3]{EKConj}. The proof of Theorem \ref{thm:polyEK} goes as follows: 
\begin{enumerate}[(i)]
    \item It is first shown that if a harmonic function $V$ is upper bounded by a polynomial of degree $n$, that is if $V (x, y) \leq A(x^2+y^2)^{n/2}$ for some $n \in \mathbb{N}, A > 0$ at large $(x,y)$, then $V$ must itself be a harmonic polynomial of degree $\leq n$. 
    \item The homogeneous, harmonic polynomials of degree $m>0$ on $\reals^2$ form a two-dimensional vector space. In the standard polar coordinates of $\reals^2$, such polynomials take the form
    \begin{equation}\label{eq:harmonicPoly}
        p_m(\rho,\theta) = \lambda_m \rho^m \cos(m(\theta + \alpha_m))
    \end{equation}    
    for $\lambda_m > 0$ and $\alpha_m \in (-\pi,\pi]$. Therefore any harmonic polynomial $P$ on $\reals^2$ of degree $n \in \mathbb{N}$ can be written as 
    \begin{equation}\label{eq:harmonicPolyFull}
    P(\rho,\theta) = \sum_{m=0}^np_m(\rho,\theta)
    \end{equation}
    for some $p_0 \in \reals$. In particular, the autonomous potential $V(q)$ of Eq. \ref{eq:autonomousV} can be written as such a sum\footnote{The extension to the non-autonomous case contains some subtleties which are explained in detail in \cite[Sec. 2.1]{EKConj}. Loosely, for a polynomially $u$-bounded and non-autonomous potential $V$, the $\lambda_m$ and $\alpha_m$ of Eq. \ref{eq:harmonicPoly} (and therefore Eq. \ref{eq:harmonicPolyFull}) become continuous functions of $u$.}. For simplicity in this summary, let us take the simple case of a homogeneous degree $n>2$ polynomial $V_n$ with $\lambda_n=-1$ and $\alpha_n=0$, that is $V_n(\rho,\theta) = -\rho^n\cos(n\theta)$. In the homogeneous case one can always obtain this via rotations, scaling or adding a real number to $V$, none of which affect the completeness or harmonic characters necessary for our discussion.
    \item Consider the radial curves in polar coordinates $\gamma_k(t) = (\rho(t),\hat{\theta}_k)$, $k \in \{0,\dots,n-1\}$ where $\hat{\theta}_k \vcentcolon= 2\pi k/n$ ($n$ is the degree of the potential $V$ being considered).
    Such curves are solutions of $\ddot{q} = -\nabla V_n(q)$ if and only if the radial component $\rho(t)$ satisfies $\ddot{\rho}(t) = n\rho^{n-1}(t)$.
    \item It is then proved that for any real number $n>2$ and $C^1$ function $\lambda:[0,\infty)\rightarrow\reals$, the solutions of the differential inequality
    \begin{equation}\label{eq:diffIneq}
    \ddot{\rho}(t) \geq n\lambda \rho^{n-1}(t)
    \end{equation}
    with initial conditions $\rho(0)>0$ and $\dot{\rho}(0) > 0$ are \textit{incomplete} under the following conditions:
    \begin{enumerate}
        \item The solutions are \textit{incomplete} if there exists some $\lambda_0 >0$ such that $\lambda \geq \lambda_0$.
        \item If $\lambda(0)>0$ then there exists some $k>0$ such that such that all solutions with initial conditions $\rho(0)>k$ or $\dot{\rho}(0)>k$ are \textit{incomplete}.
    \end{enumerate}
    The first of these points tells us immediately that the solutions $\gamma_k$ satisfying $\ddot{\rho}(t) = n\rho^{n-1}(t)$ are incomplete, as in this case $\lambda$ is the constant function equal to one, such that any $0<\lambda_0<1$ provides the necessary bound. In fact a confinement property is shown, whereby there exists regions ``around" the $\gamma_k$ labelled $D_k[\rho_0,\pi/(2n)]$ such that trajectories starting in $D_k[\rho_0,\pi/(2n)]$ (with suitable initial conditions) stay in $D_k[\rho_0,\pi/(2n)]$, and these confined solutions satisfy the differential inequality Eq. \ref{eq:diffIneq}, allowing us to prove that they too are incomplete 
    .
    \item The existence of the confining regions $D_k[\rho_0,\pi/(2n)]$ for a homogeneous potential $V_n$ can be understood as follows: Along each $\gamma_k = (\rho(t),\hat{\theta}_k = 2\pi k/n)$, $V_n(\rho,\theta) = -\rho_n\cos(n\theta)$ is decreasing and concave. Furthermore, the harmonicity\footnote{Harmonicity implies that $V_n \sim \cos(n\theta)$ such that $\frac{\partial V_n}{\partial \theta} \sim \sin(n\theta)$ and evaluating at any $\hat{\theta}_k$ yields $0$. This is the easily shown to be a minimum by taking another derivative.} of $V_n$ implies that $\frac{\partial V_n}{\partial \theta}(\gamma _k(t)) = 0 $ and that this is in fact a \textit{minimum}. That is, the $\hat{\theta}_k$ are stable equilibria of trajectories close to the $\gamma_k$. This can be visualised by looking at the potential $V_n$ for some choice of $n$. In Figure \ref{fig:homV5} the case $n = 5$ is demonstrated\footnote{Note that a very similar Figure appears in \cite[Fig. 1]{Podolsky:chaoticmotion} in a slightly different but related context, as discussed further at the end of this section.}, in which one can see $n = 5$ different ``channels" with centers corresponding to the $\gamma_k, k\in\{0,\dots,4\}$.
    
    \item To prove the case in which $V$ is not homogeneous, it is first written as a linear combination of polynomials like $V_n$. Then the $\gamma_k$ are no longer solutions of the full dynamical system $\ddot{q} = -\nabla V(q)$, but it is shown that there still exists regions ``around" the $\gamma_k$ labelled $D[\rho_0,\theta_+]$ which have qualitatively the same behaviour as the $D_k[\rho_0,\pi/(2n)]$. This is achieved by showing that the radial component of a trajectory $\gamma$ grows sufficiently fast compared to the angular oscillation that $\gamma$ never escapes the $D[\rho_0,\theta_+]$.
    \item To prove the case when $V$ is \textit{non-autonomous} a similar procedure is followed to that of the autonomous case, with some technical complications. The first notable difference is that the polar expressions of a harmonic potential $V(u,q)$ Eq. \ref{eq:harmonicPoly} and Eq. \ref{eq:harmonicPolyFull} become valid only on an interval in $u$, that is 
    \begin{equation}
        p_m(u,\rho,\theta) = \lambda_m(u) \rho^m \cos(m(\theta + \alpha_m(u))),~~~u\in (u_0 - c, u_0 + c) \subset \reals
    \end{equation}    
    for some $0 < c \in \reals$. Here we can only choose $\alpha(u_0) = 0$, and in general $\alpha(u) \neq 0$. As a result, in the non-autonomous case we have that the $\hat{\theta}_k$ are no longer constant:
    \begin{equation}
        \hat{\theta}_k(u) = \frac{2\pi k - \alpha(u)}{n},~~~k = 0,\dots,n-1.
    \end{equation}
    The remaining differences follow a similar pattern, whereby objects become $u$-dependent and are defined on intervals. However since the rough details are the same as the autonomous case, these details will be omitted here.
\end{enumerate}

\begin{figure}
    \centering
    \makebox[0.6\textwidth][c]{%
    \begin{subfigure}{.6\textwidth}
      \centering
      \includegraphics[width=1\linewidth]{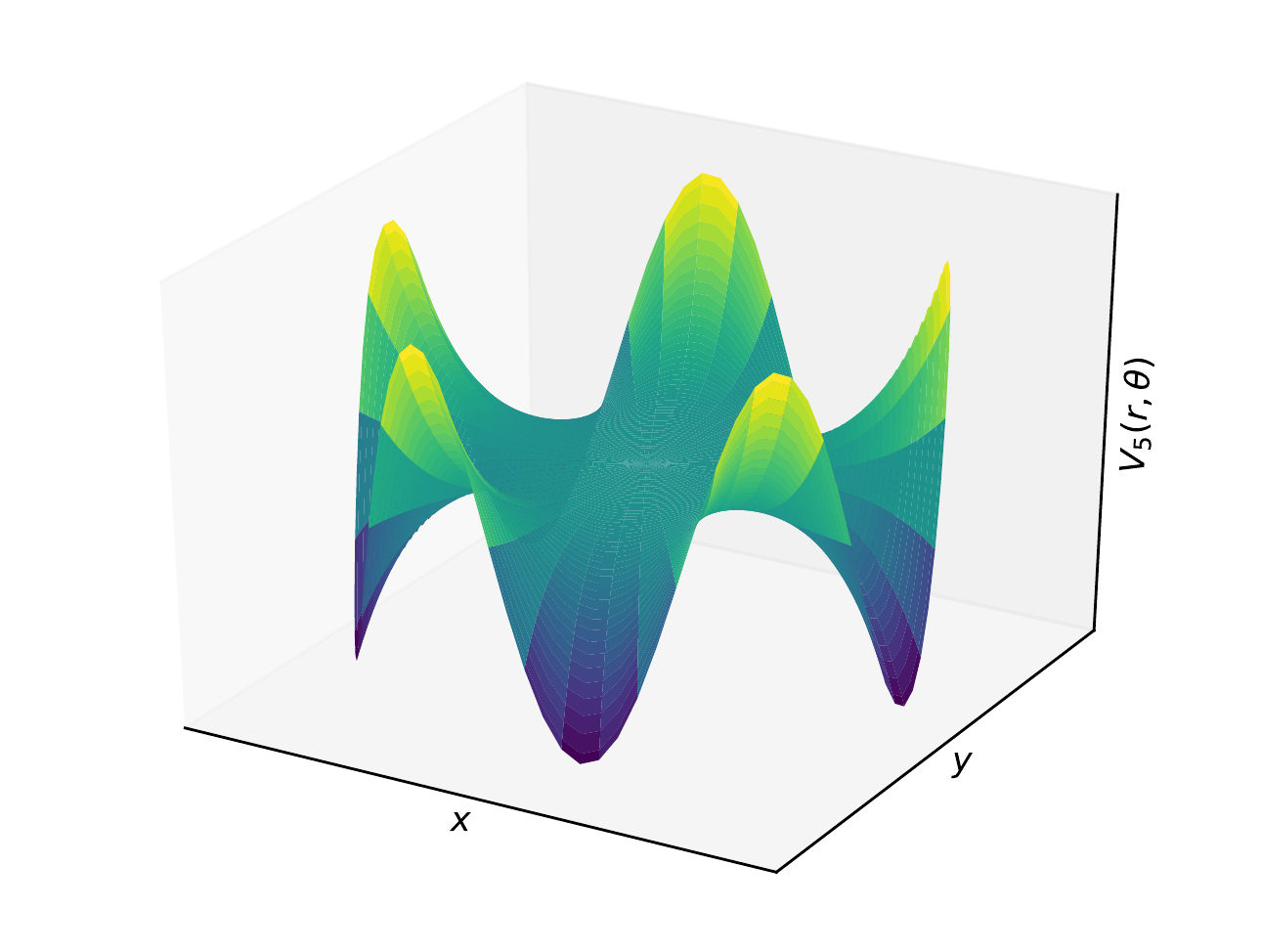}
    \end{subfigure}%
    \begin{subfigure}{.6\textwidth}
      \centering
      \includegraphics[width=1\linewidth]{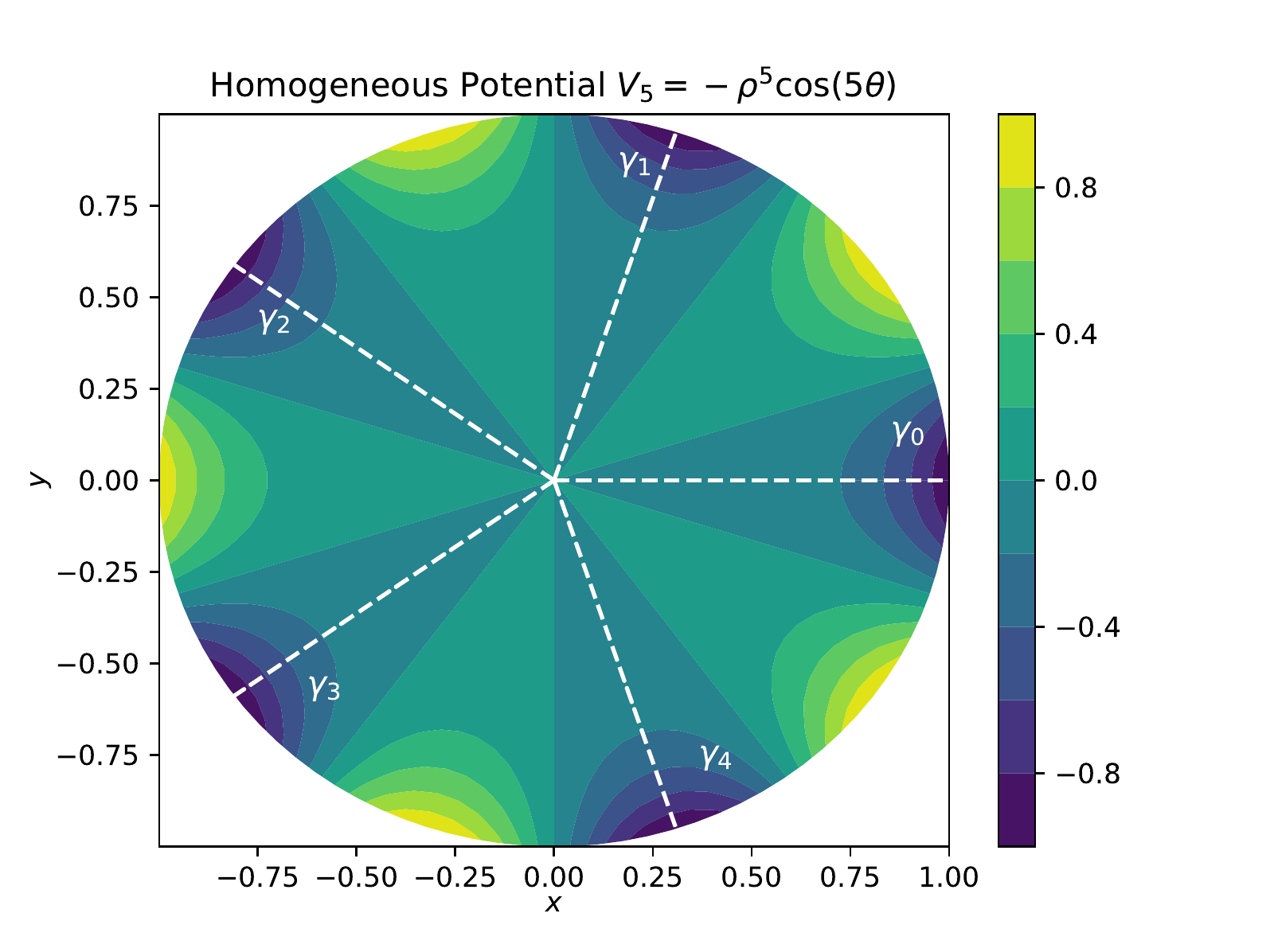}
    \end{subfigure}
    }%
    \caption{Homogeneous degree-5 potential $V_5$ as a surface (left) and contour plot (right). These images make clear the $n$ stable trajectories $\gamma_k$ for $k\in \{0,\dots,n-1\}$. The same features can be seen for any natural number $n>2$.}
    \label{fig:homV5}
\end{figure}

\textbf{Summary -- Polynomial EK-conjecture}\\
We first saw the Ehlers--Kundt conjecture, stated as:
\begin{quotation}
    \textit{``Prove the plane waves to be the only complete (gravitational) pp-waves."}
\end{quotation}
This was a statement about the completeness of the solutions of the geodesic equation for a metric $g = 2 d u d v - V(u, x,y) d u^{2}+ dx^2 + dy^2$ where $V$ is harmonic in $(x,y)$. The geodesic equations were reduced to a Hamiltonian system $\dot{p} = -\nabla V(q)$ with $q \vcentcolon= (x,y)$ and $p = \dot{q}$. In mathematical terms, the conjecture states:
\begin{equation*}
    \parbox{15em}{\centering The solutions of  $\dot{p} = -\nabla V(q)$\\exist for all times} \iff \text{$V(u,x,y)$  is quadratic  in $(x,y)$}.
\end{equation*}
The $\impliedby$ direction is already known to hold (see Sec. \ref{sec:plane}), and the $\implies$ direction is an open question. The fact that a quadratically-bounded and harmonic $V$ was proven to have complete trajectories motivated us to ask what happens if the harmonic $V$ is \textit{polynomially} bounded. This question was answered by \cite{EKConj} where it was proven that for such a $V$, all
the solutions to the dynamical system Eq. \ref{eq:dynamicalSys} are complete if and only if the function
$V (u,\cdot)$ is an at most quadratic polynomial for each $u \in \reals$. That is, the Ehlers--Kundt conjecture is proved to hold in the case that $V$ is polynomially bounded.\\

We may then ask ourselves if it is reasonable to expect that $V$ be polynomially bounded. In fact in the causal study, it was discovered that in the autonomous case unless $V$ were quadratically polynomially bounded, the pp-wave would not be strongly causal. For further evidence supporting such a bound see \cite[Sec. 13. (b)]{EKConj}. Therefore this is arguably the strongest known result addressing the EK conjecture. It is not, however, the only one; indeed, in the case of an autonomous potential, the EK conjecture has also been settled in the case when the spacetime is strongly causal, in~\cite{costa}.\\

It should also be mentioned that exactly the behavior of geodesics in geometries studied in this section (those for which $V$ is a harmonic polynomial that is even in $y$) have been studied extensively, wherein it was demonstrated via a fractal method that the geodesic flow is chaotic in nature. The geodesics escape to infinity along one of the channels which appear in Fig. \ref{fig:homV5} in this article (and in Fig. 1 of \cite{Podolsky:chaoticmotion}). For details see also \cite{Podolsky:chaos} and \cite{Vesely:chaos}. This phenomenon was further studied in the context of the sandwich waves in \cite{Podolsky:smearing}, wherein it was demonstrated that as the support of the curved region approaches zero (the so-called ``impulsive waves'') the geodesic motion becomes integrable.

\subsection{The Compact Case}
One may also wonder if the Ehlers--Kundt conjecture could be answered in the case that a pp-wave $(M,g)$ is a \textit{compact} Lorentzian manifold, since such manifolds are known to be complete under a wealth of circumstances\footnote{Though compact Lorentzian manifolds are not \textit{always} complete, in contrast to compact \textit{Riemannian} manifolds which are always complete (see Hopf–Rinow theorem \cite{hopf}).}. Some examples include when they are flat, have constant curvature, are homogeneous (and even locally homogeneous in the 3 dimensional case), or admit a time-like conformal Killing vector field \cite[pg.~2]{compact}. Unfortunately, general pp-waves do not satisfy any of these properties, and so some additional results are required to address the EK conjecture in this case. The question of completeness for compact pp-waves has indeed been answered by \cite{compact}, and that work is the subject of this section.\\

\textbf{Example}: Compact pp-wave.\\
Consider the flat metric $h$ on the n-torus $\mathbb{T}^n$, then the product manifold $M = \mathbb{T}^2 \times \mathbb{T}^n$ with the metric 
\[
    g = 2d\theta d\phi + 2H d\theta^2 + h
\]
with $H \in C^{\infty}(\mathbb{T}^n)$ is compact and is in fact a standard pp-wave with defining covariantly constant vector field represented as $\partial_{\phi}$. Note however that a ``wave" is not a very accurate name in the compact case, since as mentioned in Sec. \ref{sec:radiation} it is the (null) \textit{asymptotics} which signal the physical presence of radiation, and the compact case does not admit the same notion of ``null infinity" as was used to define the presence of radiation.\\

The principal results of \cite{compact} can be summarised as follows:

\begin{enumerate}[label=(\Alph*)]
    \item The universal cover of a compact pp-wave is globally isometric to a standard pp-wave (Eq. \ref{eq:standardPp})
    \item Every compact pp-wave $(M, g)$ is geodesically complete.
    \item Every compact Ricci-flat pp-wave is a plane wave.
\end{enumerate}

Point A is instrumental in proving point B. Point B appears to be in contradiction to the EK conjecture, but such an apparent problem is resolved by point C. That is, there are no non-plane compact vacuum pp-waves, so we need not wonder about their completeness on physical grounds. Thus these results solve the Ehlers--Kundt conjecture in the compact case. Or rather, the authors have proven that one need not conjecture about the incompleteness of non-plane vacuum compact pp-waves, as there are no such pp-waves. The remainder of this section will outline the methods by which these results are obtained. Let us begin with result (A) in more detail:
\begin{theorem}\label{thm:cover}
    The universal cover of an $n$-dimensional\footnote{Note that the authors of the original work \cite{compact} use $n$ as the dimension of only the wavefront, and in this article it is the dimension of the spacetime. Therefore $n_{\text{this article}} = n_{\text{Leistner et al.}} + 2$. Similarly, a different convention on $H$ is used, such that the $H_{\text{this article}} = 2H_{\text{Leistner et al}}$. This does not impact the methods used in any meaningful way.}
    compact pp-wave defined by a covariantly constant null vector field $Z$ is globally
    isometric to a standard pp-wave (Eq. \ref{eq:standardPp}) which can be written as
    \[
        (\reals^{n}, g^H = 2 d u d v+H\left(u, \mathbf{x}\right) d u^{2}+\delta _{a b}d x^{a} d x^{b} )
    \]
    and under this isometry, the lift of $Z$ is mapped to the coordinate vector field $\frac{\partial}{\partial v}$
\end{theorem}

Though we don't present the proof of this theorem here, we remark that it makes significant use of the ``screen bundle" which is closely related to the ``wavefront" of our Definition \ref{def:wavefront}. However, as remarked in \cite[footnote 2]{compact} in the compact case this nomenclature is perhaps inappropriate. Using Theorem \ref{thm:cover}, it is then proven that:
\begin{theorem}\label{thm:compactComplete}
    Every compact pp-wave $(M, g)$ is geodesically complete.
\end{theorem}

To prove this statement, let us first examine the completeness of a standard pp-wave (Eq. \ref{eq:standardPp}). Then via Theorem \ref{thm:cover} we can make statements about the completeness of compact pp-waves. Recall that a standard pp-wave may be written in the global coordinate chart $\{u,v,x^1,\dots,x^{n-2}\}$ as
\begin{equation}
g = 2 d u d v+H\left(u, \mathbf{x}\right) d u^{2}+\delta _{a b}d x^{a} d x^{b}.
\end{equation}

\begin{proposition}\cite[lemma 8]{compact}
The standard pp-wave metric is geodesically complete if
\[
    \left|\frac{\partial^{2} H}{\partial x^{i} \partial x^{j}}\right| \leq c
\]
for $0 < c \in \reals$ for all $i,j \in \{1\dots,n-2\}$
\end{proposition}
\begin{proof}
    Let us examine the geodesic equations of the standard pp-wave metric: For a curve $\gamma$ with components $(u(s), v(s), x^1(s), \dots, x^{n-2}(s))$, the geodesic equation for the $u$-component is given by:
    \[
        \ddot{u}(s)  =0  \Longrightarrow u(s)=a s+b~\text{for some }a,b \in \reals\\
    \]
    that is, the $u$ component is defined on all of $\reals$. The remaining components of the geodesic equations are given by
    \begin{align}
    \ddot{v}(s)  &=-2 a \dot{x}^{k}(s) \frac{\partial H}{\partial x^{k}}-a^{2} \frac{\partial H}{\partial u},\\
    \ddot{x}^{k}(s)  &=a^{2} \frac{\partial H}{\partial x^{k}}. 
    \end{align}
    Since the $v$ equation only depends on the $x^k$ and not on $v$, then the solution is defined on $\reals$ provided that the $x^{k}$ are defined on $\reals$. Unfortunately the $x^k$ equation does not in general admit solutions on all of $\reals$. An example (as in \cite{LeistnerTalk})
    is found when $H = \frac{1}{2}(x^j)^4$ for some $j\in\{1,\dots,n-2\}$. In this case, the only nontrivial equation (when $a\neq 0$) for the $x^k(s)$ is
    \[
        \ddot{x}^j(s) = 2a^2(x^j)^3
    \]
    which has solution 
    \[
        x^j(s) = \frac{1}{1-as}, ~~s\in (-\infty,1/a).
    \]
    Since this solution develops a singularity, so too does the solution for $v$, and we conclude that the standard pp-wave is geodesically incomplete in this case. So then when are the solutions of the $\ddot{x}^k$ equations defined on all of $\reals$ (thus making the pp-wave geodesically complete)? This is guaranteed when the second derivatives of $H$ are bounded; as then by the mean value theorem the first derivatives are Lipschitz continuous which suffices in view of the Picard--Lindel\"of theorem.
    
\end{proof}

One may think that this result yields many examples of complete pp-waves which are non-plane (and are instead just bounded in second derivative of $H$) but in fact we have not imposed that the pp-wave is \textit{gravitational}. For a gravitational pp-wave $H$ is harmonic, and a harmonic function can only have bounded second derivatives (corresponding to a complete pp-wave by the previous proposition) if it is quadratic and thus a plane wave\footnote{Note that this is the content of Remark 5 of the original work \cite{compact}. Their Remark 5 concludes with ``thus a pp-wave", but this should in fact read ``thus a plane wave". The correct conclusion is reached in this article, and we thank Prof. Leistner for confirming.}.\\

In order to apply this result to our case, that is to prove that a compact pp-wave is geodesically complete (Theorem \ref{thm:compactComplete}), we must prove that the second derivatives of $H$ are bounded in the compact case. The following proposition resolves this question:
\begin{proposition}\label{prop:boundedDerivs}
    Consider a compact pp-wave. By Theorem \ref{thm:cover}, its universal cover is a standard pp-wave $(\reals^{n}, g = 2 d u d v+H\left(u, \mathbf{x}\right) d u^{2}+\delta _{a b}d x^{a} d x^{b} )$. Then the second derivatives of $H$ are bounded
    \[
        0 \leq \frac{\partial^{2} H}{\partial x^{i} \partial x^{j}} \leq c ~~~\forall ~i,j = 1,\dots,n-2.
    \]
\end{proposition}
\begin{proof}
    We again omit the proof in favour of brevity. See \cite[lemma 9]{compact}. 
\end{proof}

Thus one arrives at a \textbf{proof of theorem (B)}:
\begin{proof}
    Let $(M,g)$ be a compact pp-wave. By theorem (A) the universal cover is isometric to a standard pp-wave, and by the above proposition such a standard pp-wave is complete. Therefore $(M,g)$ itself is complete.
\end{proof}

We finally arrive at the statement which resolves the EK conjecture in the case of compact pp-waves.
\begin{theorem} \cite[Corollary 1]{compact}
    Every compact Ricci-flat pp-wave is a plane wave\footnote{Note that there are examples of compact non-plane pp-waves, but they are not Ricci-flat.}.
\end{theorem}
\begin{proof}
    Let $(M, g)$ be a compact pp-wave and let $(\reals^{n+2}, g^H)$ be the standard pp-wave that is globally isometric to the universal cover of $(M, g)$. As in Proposition \ref{prop:boundedDerivs}, we have that the second derivatives of $H$ are bounded. If $g$ is Ricci-flat, so too is $g^H$ , and thus $H$ is harmonic with respect to the $x^i$ directions 
    \[
        \sum_{i=1}^{n-2} \partial_i^2 H = 0.
    \]
    But this implies that also $\partial_i \partial_j H$ is
    harmonic in the same sense, and thus, by the maximum principle for harmonic functions \cite[page 7]{harmonic}, independent of the $x^i$ components. Hence,
    \[
        H(u,\mathbf{x})=\sum_{i, j=1}^{n-2} a_{i j}(u) x^{i} x^{j}+b_{i}(u) x^{i}+c(u),
    \]
    where $a_{ij}$, $b_i$ and $c$ depend only on $u$ and not the $x^i$, and thus since $H$ is quadratic in $x^i$, $(M,g)$ is a plane wave.
\end{proof}

Therefore as stated, one need not conjecture about the incompleteness of non-plane vacuum compact pp-waves, as there are no such pp-waves. As a result, the Ehlers--Kundt conjecture has been resolved in the compact case.

\subsection{Case of Failure}
Let us outline very briefly the following case in which the Ehlers--Kundt conjecture is known \textit{not} to hold:\\

\textbf{Impulsive case}:\\
Though usually omitted for brevity in this article, the continuity of the characteristic function $H$ of a pp-wave in $u$ of the adapted coordinates is in fact vital. To quote from \cite[Sec. 1.3 (d)]{EKConj}:
\begin{quote}
    \textit{Impulsive waves have a non-continuous profile type
$H(u,z=(x,y)) = f(z)\delta(u)$ for some (generalized) delta-function $\delta$ and smooth $f$. Thus, the function $H$ can be regarded as $z$-harmonic when $\Delta f = 0$. The mentioned results of completeness yield counterexamples to the EK conjecture in the impulsive setting, showing the necessity of continuity in $u$ as well as the appropriate smoothness of $H$.}
\end{quote}

This necessary smoothness and continuity in the non-autonomous case ($H$ not independent of $u$) amounts to 
\begin{itemize}
    \item $H$ should be $C^1$ in $u$ (for constructing Levi-Civita Connection)
    \item $H$ should be $C^2$ in $z$ (to impose harmonicity, i.e. vacuum condition)
\end{itemize}

(Note that, in the second condition, being $C^2$ in $z$ is equivalent to being analytic in $z$, a well known property of harmonic functions (see, e.g., \cite[Theorem~1.28]{Axler}).  For the relevant references in the study of such impulsive waves, consult \cite[Sec. 1.3 (d)]{EKConj}.

%% file: Postamble/appendix.tex
\section{Proof of Theorem \ref{thm:adapted}}\label{app:proof}
Note the same notation as described at the beginning of Section \ref{sec:coordDesc} will be used throughout this proof. We also reproduce the statement of the theorem here for completeness.

\begin{theorem}\textup{Coordinates adapted to a covariantly constant\footnote{Note that for this particular result, one may relax the condition that $Z$ be covariantly constant. For details see Sec. \ref{sec:penroseLimit}. In this context however, $Z$ is always assumed to be covariantly constant.} null vector field.}\\
    If a Lorentzian manifold $(M,g)$ admits a covariantly constant, null vector field $Z$, then in a neighbourhood $U$ of each $p \in M$ there exists a local coordinate chart $\varphi = \{u,v,\mathbf{x}\}$ on $U$ which is ``adapted to $Z$" such that 
     $$Z\rvert_U = \partial_v = \nabla u.$$
    \end{theorem}

\begin{proof}
    We perform this proof in three steps: first we construct local coordinates in which $Z$ is a coordinate vector field, then we show there exists a function $u:M\mapsto \reals$ such that $Z = \text{grad}(u) = \nabla u$. Finally we show that such a function $u$ may replace one of the coordinate functions in the initially constructed system, while maintaining the property $Z = \partial_v$, giving the desired result.\\
    
    \textbf{Step 1: Local coordinates including $\mathbf{v}$}\\
        First note that $Z$ is nowhere $0$. This is because the zero vector is by convention spacelike, but by definition $Z$ is everywhere null.  This can also be seen from the fact that $Z$ is Killing, since the Killing condition is trivially satisfied. A Killing field is uniquely determined by $Z\rvert _p$ and $\nabla Z \rvert_p$ for some $p \in M$. Therefore since $\nabla Z = 0$ everywhere, if for some $p$ we had $Z\rvert _p = 0$ then $Z$ would be identically 0.\\
            
        By the straightening (or ``flow-box", or ``canonical form") theorem \cite[Theorem 2.4.3]{flowbox} for vector fields, since $Z$ is everywhere regular (i.e. nonzero), we can always construct a local coordinate system such that $Z$ is a coordinate vector field. We label this coordinate $v$, such that we have $Z = \tilde{\partial}_v$ in the local coordinates $\{\tilde{x}^0,v,\tilde{x}^1,\ldots,\tilde{x}^{n-2}\}$. We define these coordinates with a tilde because we are interested in constructing a new coordinate system from these coordinates, and $v$ is the second coordinate to agree with usual pp-wave notational conventions. Let us also write the coordinate vector fields for this coordinate system with a tilde as $\tilde{\partial}_v$ and $\tilde{\partial}_i$.\\
        
    \textbf{Step 2: Introducing the coordinate $\mathbf{u}$}\\
        To obtain the function $u$, let us consider the one-form $Z^\flat$ dual to $Z$ via the metric $g$. As shown in Appendix \ref{app:extderiv}, the exterior derivative $d\omega$ of a one-form $\omega$ is proportional to $Alt(\nabla \omega)$, the antisymmetric part of the two-form $\nabla \omega$. Since $Z$ is covariantly constant, by the compatibility of the metric with $\nabla$ we have that $\nabla Z^\flat = 0$  and thus $Alt(\nabla Z^\flat) = 0$. Therefore $d(Z^\flat) = 0$, that is $Z^\flat$ is \textit{closed}, and via the Poincar\'{e} lemma for covector fields\footnote{For details see Lee, Smooth manifolds (2nd edition) Theorem 11.49 and Corollary 11.50.}, any closed one-form can locally be written as $Z^\flat = du$ for some function $u: M \rightarrow \reals$. Then by definition of the gradient, we have $Z = grad(u) = \nabla u$.\\
        
    \textbf{Step 3: Constructing the local coordinate system $\{u,v,\mathbf{x}\}$}\\
        We now transform the coordinate system $\{\tilde{x}^0,v,\tilde{x}^1,\ldots,\tilde{x}^{n-2}\}$ into a new coordinate system $\{u,v,x^1,\ldots,x^{n-2}\}$, and show that the property $Z = \partial_v$ also holds in the new coordinates
        \footnote{It is necessary to show this even though the function $v$ is used in both coordinate systems, as given $n$ functions $\tilde{f}^i$ with linearly independent differentials $\tilde{df}^i$, we may form a local coordinate system $\{\tilde{f}^1,\ldots,\tilde{f}^n\}$ in which the coordinate vector fields $\tilde{\partial_i}$ are determined by the $n^2$ equations $\tilde{df}^j(\tilde{\partial_i})=\delta^j_i$. That is, the coordinate vector field $\tilde{\partial_{i}}$ depends on all the coordinate functions. If we transform to a new coordinate system $\{f^1,\ldots,f^{i-1},\tilde{f}^i,f^{i+1},\ldots,f^n\}$ which contains $\tilde{f}^i$, we have $\partial_i = \tilde{\partial_i} \iff df^j(\tilde{\partial_i}) = \delta^j_i$ for all $j$.}.
        To find the appropriate transformation, first consider $du$ acting on the coordinate vector fields $\tilde{\partial_i}$. We have $du(\tilde{\partial_v}) = Z^\flat(Z) = g(Z,Z) = 0$ and we define $du(\tilde{\partial_i}) =\vcentcolon c_i$ where the $c_i$ are smooth functions on $M$. Since the coordinate vector fields form a frame, at any $p \in M$ we cannot have $c_i = 0$ for all $i$ and $c_0=0$, as if this were true we would have $du(X) = 0$ for all vector fields $X$, which in turn implies $du = Z^\flat = 0$. But since $Z$ is nonzero, so too is $Z^\flat$. Without loss of generality, assume that $c_0 \neq 0$ (can always be done by reordering/relabelling the coordinate system).\\
        
        We now claim that by replacing $\tilde{x}^0$ by $u$ and taking $x^i = \tilde{x}^i$ for $i \in \{1,\ldots,n-2\}$, the set of functions $\{u,v,x^1,\ldots,x^{n-2}\}$ form a valid coordinate system. This can be achieved by verifying that the Jacobian $J$ of the coordinate transform is invertible. This is easily seen from the fact that $J$ in Eq. \ref{eq:Jacobian} (where $\mathds{1}_{n-1}$ is the identity matrix in $n-1$ dimensions) has linearly independent columns for $c_0 \neq 0$.
        
        \begin{equation}\label{eq:Jacobian}
        J = 
        \begin{pmatrix}
            du(\tilde{\partial_0})  & dv(\tilde{\partial_0})    & dx^1(\tilde{\partial_0})  & \dots  & dx^{n-2}(\tilde{\partial_0})        \\
            du(\tilde{\partial_v})  & dv(\tilde{\partial_v})    & dx^1(\tilde{\partial_v})  & \dots  & dx^{n-2}(\tilde{\partial_v})        \\
            du(\tilde{\partial_1})  & dv(\tilde{\partial_1})    & dx^1(\tilde{\partial_1})  & \dots  & dx^{n-2}(\tilde{\partial_1})        \\
            \vdots                  & \vdots                    & \vdots                    & \ddots & \vdots         \\
            du(\tilde{\partial}_{n-2})  & dv(\tilde{\partial}_{n-2})    & dx^1(\tilde{\partial}_{n-2})  & \dots  & dx^{n-2}(\tilde{\partial}_{n-2})
        \end{pmatrix}
        =
        \begin{pNiceMatrix}[columns-width = 4mm]
            c_0     & 0         & \Cdots        &           & 0        \\
            0       & \Block{4-4}<\Large>{\mathds{1}_{n-1}}         &             &     &         \\
            c_1     &          &         &           &         \\
            \vdots  &     &               &           &          \\
            c_{n-2}     &          &        &         & 
        \end{pNiceMatrix}
        \end{equation}
        
        It remains to show that $Z = \partial_v$ in the new coordinates. That is, we must have 
        $du(Z) = dx^i(Z) = 0$ for all $i \in \{1,\dots,n-2\}$ and $dv(Z) = 1$. This however can be read directly from the second row of the Jacobian (Equation \ref{eq:Jacobian}). We therefore have that $\{u,v,\mathbf{x}\} \vcentcolon = \{u,v,x^1,\dots,x^{n-2}\}$ is a valid local coordinate system on $M$.
    \end{proof}
    
\section{Exterior derivative of $k$-forms}\label{app:extderiv}
The following is based on \cite{jackLee}.
There are two primary conventions for defining a wedge product, which are in fact proportional to each other. The first is that which is used in \cite{spivak}: for $\alpha$ a $k$-form and $\beta$ an $l$-form
\begin{equation}
    \alpha \wedge \beta=\frac{(k+l) !}{k ! l !} \operatorname{Alt}(\alpha \otimes \beta).
\end{equation}
The second is that of \cite{nomizu}, and is given by 
\begin{equation}
    \alpha \wedge \beta=\operatorname{Alt}(\alpha \otimes \beta)
\end{equation}
In both cases, the wedge product is proportional to $\operatorname{Alt}(\alpha \otimes \beta)$. Let us choose convention 1 and write explicitly $\operatorname{Alt}(\nabla \omega)$ for a $k$-form $\omega$
\begin{equation}
    \operatorname{Alt}(\nabla \omega)\left(X_{1}, \cdots, X_{k+1}\right)=\frac{1}{(k+1) !} \sum_{\sigma \in S_{k+1}}(\operatorname{sgn} \sigma) \nabla \omega\left(X_{\sigma(1)}, \cdots, X_{\sigma(k+1)}\right)
\end{equation}
for smooth vector fields $X_j$. The exterior derivative of a $k$-form $\omega$ is given by
\begin{equation}
    \begin{array}{c}
\mathrm{d} \omega\left(X_{1}, \cdots, X_{k+1}\right)=\sum_{i=1}^{k}(-1)^{i+1} X_{i}\left(\omega\left(X_{1}, \cdots, \widehat{X}_{i}, \cdots, X_{k+1}\right)\right)+ \\
\sum_{i<j}(-1)^{i+j} \omega\left(\left[X_{i}, X_{j}\right], X_{1}, \cdots, \widehat{X}_{i}, \cdots, \widehat{X}_{j}, \cdots, X_{k+1}\right),
\end{array}
\end{equation}
    where the $\widehat{X}_j$ denotes that the argument $X_j$ is to be omitted. 
By taking an alternating product proportional to the wedge product with any constant of proportionality (i.e. convention) and comparing both the $d\omega$ and $Alt(\nabla \omega)$ in Riemann normal coordinates (as they are both tensors and thus can be compared pointwise), one finds that the expressions simplify greatly and are indeed proportional to each other.